\documentclass[aoas,preprint]{imsart}
\RequirePackage[OT1]{fontenc}
\RequirePackage{amsthm,amsmath}

\RequirePackage[colorlinks,citecolor=blue,urlcolor=blue]{hyperref}

\theoremstyle{plain}

\numberwithin{equation}{section}

\input{bps_AOAS.sty}

\begin{document}
\begin{frontmatter}
\title{Large-scale inference of correlation among mixed-type biological traits with phylogenetic multivariate probit models}
\runtitle{Phylogenetic multivariate probit models}

\begin{aug}

	\author{\fnms{Zhenyu} \snm{Zhang}\thanksref{m1}\ead[label=e1]{zyz606@ucla.edu}},
	\author{\fnms{Akihiko} \snm{Nishimura}\thanksref{m7}\ead[label=e2]{akihiko4@gmail.com}},
	\author{\fnms{Paul} \snm{Bastide}\thanksref{m2,m9}\ead[label=e3]{paul.bastide@umontpellier.fr}},
	\author{\fnms{Xiang}\
	\snm{Ji}\thanksref{m8}\ead[label=e4]{xji4@tulane.edu}},
	\author{\fnms{Rebecca P.}\
	\snm{Payne}\thanksref{m3}\ead[label=e5]{rebecca.payne2@ncl.ac.uk}},
	\author{\fnms{Philip}\
	\snm{Goulder}\thanksref{m4,m5,m6}\ead[label=e6]{philip.goulder@paediatrics.ox.ac.uk}},
	\author{\fnms{Philippe}
	\snm{Lemey}\thanksref{m2}\ead[label=e7]{philippe.lemey@kuleuven.be}}
	\and
	\author{\fnms{Marc A.} \snm{Suchard}\thanksref{m1} \corref{} \ead[label=e8]{msuchard@ucla.edu}}

	\runauthor{Z.~Zhang et al.}

	\affiliation{University of California, Los Angeles\thanksmark{m1}, Johns Hopkins University \thanksmark{m7}, KU Leuven\thanksmark{m2}, Universit\'e de Montpellier\thanksmark{m9}, Tulane University \thanksmark{m8}, \mbox{Newcastle University\thanksmark{m3}}, University of Oxford\thanksmark{m4}, University of KwaZulu-Natal\thanksmark{m5}, and Ragon Institute of MGH, MIT, and Harvard\thanksmark{m6}}

		\address{
		Z. Zhang\\
		Department of Biostatistics\\
		Fielding School of Public Health\\
		University of California, Los Angeles\\
		Los Angeles, California 90095-1772\\
		USA\\
		\printead{e1}
	}

	\address{
		A. Nishimura\\
		Department of Biostatistics\\
		Bloomberg School of Public Health\\
		Johns Hopkins University\\
		615 N. Wolfe Street\\
		Baltimore, MD 21205\\
		USA\\
		\printead{e2}\\
	}

	\address{
		P. Bastide\\
        IMAG\\
        Universit\'e de Montpellier\\
        CNRS\\
        Place Eug\`ene Bataillon\\
        34090 Montpellier\\
        France\\
		\printead{e3}\\
	}

	\address{
		X. Ji\\
		Department of Mathematics\\
		School of Science \& Engineering\\
		Tulane University\\
		6823 St. Charles Avenue\\
		New Orleans, Louisiana 70118-5698\\
		USA\\
		\printead{e4}\\
    }

	\address{
	R. P. Payne\\
	Translational and Clinical Research Institute\\
	Newcastle University  \\
	Newcastle upon Tyne, NE2 4HH\\
	UK\\
	\printead{e5}
}

	\address{
	P. Goulder\\
	Department of Paediatrics\\
	University of Oxford \\
	Headington, Oxford, OX3 9DU\\
	UK\\
	HIV Pathogenesis Programme \\
	Doris Duke Medical Research Institute \\
	University of KwaZulu-Natal\\
	238 Mazisi Kunene Rd, Glenwood, Durban\\
	South Africa\\
	Ragon Institute of MGH, MIT, and Harvard \\
	Cambridge, MA 02139-3583\\
	USA\\
	\printead{e6}
}

	\address{
		P. Lemey\\
		Department of Microbiology, \\
		\, Immunology and Transplantation\\
		Rega Institute\\
		KU Leuven\\
		Herestraat 49\\
		3000 Leuven\\
		Belgium\\
		\printead{e7}\\
	}

	\address{
		M. A. Suchard\\
		Department of Biomathematics,\\
		\, Biostatistics and human genetics\\
		University of California, Los Angeles\\
		6558 Gonda Building\\
		695 CHARLES E. Young Drive\\
		Los Angeles, California 90095-1766\\
		USA\\
		\printead{e8}
	}
\end{aug}

\begin{abstract}
Inferring concerted changes among biological traits along an evolutionary history remains an important yet challenging problem.
Besides adjusting for spurious correlation induced from the shared history, the task also requires sufficient flexibility and computational efficiency to incorporate multiple continuous and discrete traits as data size increases.
To accomplish this, we jointly model mixed-type traits by assuming latent parameters for binary outcome dimensions at the tips of an unknown tree informed by molecular sequences. This gives rise to a phylogenetic multivariate probit model.
With large sample sizes, posterior computation under this model is problematic, as it requires repeated sampling from a high-dimensional truncated normal distribution.
Current best practices employ multiple-try rejection sampling that suffers from slow-mixing and a computational cost that scales quadratically in sample size.
We develop a new inference approach that exploits 1) the bouncy particle sampler (BPS) based on piecewise deterministic Markov processes to simultaneously sample all truncated normal dimensions, and 2) novel dynamic programming that reduces the cost of likelihood and gradient evaluations for BPS 
to linear in sample size. In an application with 535 HIV viruses and 24 traits that necessitates sampling from a 12,840-dimensional truncated normal, our method makes it possible to estimate the across-trait correlation and detect factors that affect the pathogen's capacity to cause disease.
This inference framework is also applicable to a broader class of covariance structures beyond comparative biology.
\end{abstract}


\end{frontmatter}

\section{Introduction}

Phylogenetics stands as a key tool in assessing rapidly evolving pathogen diversity and its impact on human disease.
Important taxonic examples include RNA viruses, such as influenza and human immunodeficiency virus (HIV).
Pathogens sampled from infected individuals are implicitly correlated with each other through their shared evolutionary history, often described through a phylogenetic tree that one reconstructs by sequencing the pathogen genomes.
Drawing inference about concerted changes within multiple measured pathogen and host traits along this history leads to highly structured models.
These models must simultaneously entertain and adjust for the across-taxon correlation and the between-trait correlation that characterizes the trait evolutionary process, leading to high computational burden.
This burden arises from the need to integrate over the unobserved trait process and possible uncertainty in the history.
This burden grows more challenging as the sample size, both in terms of number of taxa $\nTaxa$ and number of traits $\nTraits$, increases and, especially, when traits are of mixed-type, including both continuous quantities and discrete outcomes.
Here, even best current practices \citep{Cybis2015} fail to provide reliable estimates for emerging biological problems due to high computational complexity.

To jointly model continuous and binary trait evolution along an unknown tree, we adopt and extend the popular phylogenetic threshold model for binary traits \citep{felsenstein2005, felsenstein2011cmp} with a long tradition in statistical genetics \citep{wright1934analysis}.
This model assumes that unobserved continuous latent parameters for each tip taxon in the tree determine the observed binary traits according to a threshold.
The latent parameters themselves arise from a Brownian diffusion along the tree \citep{felsenstein1985phylogenies}. The correlation matrix of the diffusion process informs correlation between latent parameters that map to concerted changes between binary traits. Here one interprets the latent parameters as the combined effect of all relevant genetic factors that influence the binary traits after adjusting for the shared evolutionary history.

As in \citet{Cybis2015}, we extend the threshold model to include continuous traits by treating them as directly observed dimensions of the latent parameters. We recognize an identifiability issue in \citet{Cybis2015} and address this limitation with specific constraints on the diffusion covariance. We arrive at a mixed-type generalization of the  multivariate probit model \citep{chib1998} that allows us to jointly model continuous and binary traits. We call this the phylogenetic multivariate probit model. Similar strategies for mixed-type data that assume latent processes underlying discrete data are commonly employed in various domain fields, including the biological and ecological sciences \citep{schliep2013multilevel, irvine2016extending,clark2017generalized}, optimal design \citep{fedorov2012optimal}, and computer experiments \citep{pourmohamad2016multivariate}. The observed outcomes can also be conveniently clustered \citep{dunson2000bayesian, murray2013bayesian}.
Likewise, our phylogenetic probit model is easily extendable to categorical and ordinal data \citep{Cybis2015}.

Alternative approaches for mixed-type traits on unknown trees are limited.
Phylogenetic regression models \citep{grafen1989phyreg} assume a known fixed tree and their logistic extensions \citep{ives2009phylogenetic} take a single binary trait as the regression outcome.
On the other hand, for continuous traits, comparative methods \citep{felsenstein1985phylogenies} scale well on random trees \citep{pybus2012,tung2014linear}.  Likewise, continuous-time Markov chain based methods \citep{pagel1994detecting,lewis2001likelihood} are popular for multiple binary traits, but restrictively assume independence between traits given the tree.

Bayesian inference for the phylogenetic multivariate probit model involves, however, repeatedly sampling latent parameters from an $\nTaxa \nTraits$ dimensional truncated normal distribution, with $ \nTaxa $ being the number of taxa and $ \nTraits $ the number of traits.  To attempt this, \cite{Cybis2015} use Markov chain Monte Carlo (MCMC) based on a multiple-try rejection sampler.
The sampler has a computational complexity of $\order{\nTaxa \nTraits^2}$ to update $ \nTraits $ dimensions of the latent parameters for just one taxon within a Gibbs cycle.
Hence, to touch all dimensions, the resulting cost is $ \order{\nTaxa^2 \nTraits^2} $.
Further,
since only a small portion of the latent parameter dimensions are updated per rejection-sample, the resulting MCMC chain is highly auto-correlated, hurting efficiency.

To overcome this limitation, we develop a scalable approach to sample from the multivariate truncated normal by combining the recently developed bouncy particle sampler (BPS) \citep{bouchard2018bouncy} and an extension of the dynamic programming strategy by \citet{pybus2012}.
BPS samples from a target distribution by simulating a Markov process with a piecewise linear trajectory.
The simulation generally requires solving a one-dimensional optimization problem within each line segment.
When sampling from a truncated normal, however, this optimization problem can be solved via a single log-density gradient evaluation.
In the phylogenetic multivariate probit model, a direct evaluation of this gradient requires $\order{\nTaxa^2 \nTraits + \nTaxa \nTraits^2}$ computation.
By extending the dynamic programming strategy of \citet{pybus2012} for diffusion processes on trees, we reduce this computational cost to $ \order{\nTaxa \nTraits^2} $ --- a major practical gain as $\nTaxa \gg \nTraits$ in most applications. Compared to the current practice, our BPS sampler achieves superior mixing rate, allowing us to attack previously unworkable problems.

We apply this Bayesian inference framework to assess correlation between HIV-1 \textit{gag} gene immune-escape mutations and viral virulence, the pathogen's capacity to cause disease. By adjusting for the unknown evolutionary history that confounds our epidemiologically collected data, we identify significant correlations that closely match with the biological experimental literature and increase our understanding of the underlying molecular mechanisms of HIV.

\section{Modeling}
\label{sec:model}

\subsection{Phylogenetic multivariate probit model for mixed-type traits}\label{sec:model1}
Consider $ \nTaxa $ biological taxa, each with $\nTraits$ trait measurements. These measurements partition as $ \observedResponse = \observedResponseMatrix= \left[\observedDiscrete, \observedCont\right]$ with $ \observedDiscrete $ being an $ \nTaxa \times \nTraitsDiscrete $ matrix of $\nTraitsDiscrete$ binary traits and $ \observedCont $  an $ \nTaxa \times \nTraitsCont $ matrix of $\nTraitsCont$ continuous traits, where $ \nTraits = \nTraitsDiscrete + \nTraitsCont$. We assume that $ \observedResponse $ arises from a partially observed multivariate Brownian diffusion process along a phylogenetic tree $ \phylogeny $. The tree $ \phylogeny = (\nodeSet, \branchSet) $ is a directed, bifurcating acyclic graph with a set of nodes $ \nodeSet $ and branch lengths $ \branchSet $. The node set $ \nodeSet $ contains $ \nTaxa $ degree-1 tip nodes, $ \nTaxa - 2 $ internal nodes of degree 3, and one root node of degree 2. The branch lengths $ \branchSet = \left( \branchLength{1}, \dots,   \branchLength{2\nTaxa - 2}\right) $ denote the distance in real time from each node to its parent (Figure \ref{fig:vmatrix}, left).
The tree $ \phylogeny $ is either known or informed by molecular sequence alignment $ \sequenceData $ \citep{beast2018}.

We associate each node $ \nodeIndexOne$ in $ \phylogeny $ with a latent parameter $ \latentData_\nodeIndexOne \in \realNumbers^{\nTraits}$ for $i = 1, \dots, 2\nTaxa - 1$. A Brownian diffusion process characterizes the evolutionary relationship between latent parameters, such that
$ \latentData_{\nodeIndexOne}$ is multivariate normal (MVN) distributed,
\begin{equation}
\latentData_{\nodeIndexOne} \sim \normalDistribution{\latentData_{\text{pa}\left(\nodeIndexOne\right)}}{\branchLength{\nodeIndexOne}\traitCovariance},
\end{equation}
centered at its parent node value $\latentData_{\text{pa}\left(\nodeIndexOne\right)}$ with across-trait, per-unit-time, $ \nTraits \times \nTraits $ variance matrix $ \traitCovariance $ that is shared by all branches along $\phylogeny$.

At the tips of $ \phylogeny $, we collect the $ \nTaxa \times \nTraits$ matrix  $ \latentData = \latentDataMatrix = \left[\latentData_{1}, \dots, \latentData_\nTaxa \right]\transpose $ and map it to the observed traits through the function
\begin{equation}\label{eq:thresholdFunc}
\observedResponseElement = \mappingFunc{\latentDataMatrixElement}=
\begin{cases}
\sign(\latentDataMatrixElement), & j = 1, \dots \nTraitsDiscrete ,\\
\latentDataMatrixElement, & j = \nTraitsDiscrete + 1, \dots, \nTraits,
\end{cases}
\end{equation}
where $ \sign(\latentDataMatrixElement)$ takes the value 1 on positive values and -1 on negative values. As a result, latent parameters at the tips and a threshold (that we set to zero without loss of generality) determine the corresponding binary traits, and continuous traits can be seen as directly observed.

Turning our attention to the joint distribution of tip latent parameters $ \latentData $, we can integrate out $ \latentData_{\nTaxa + 1}, \dots, \latentData_{2\nTaxa - 1} $ by assuming a conjugate prior on the tree root, $ \latentData_{2\nTaxa - 1} \sim \normalDistribution{\rootpriorMean}{\rootpriorSamplesize \inverse \traitCovariance}$ with prior mean $\rootpriorMean$ and prior sample size $\rootpriorSamplesize$.
Then $ \latentData $ follows a matrix normal  distribution
\begin{equation}\label{eq:matrixNormalDistri}
\latentData \sim \matrixNormal{\nTaxa}{\nTraits}{\tipMeanMatrix}{\phylogenyVariance}{\traitCovariance},
\end{equation}
where $ \tipMeanMatrix = \left(\rootpriorMean,\dots,\rootpriorMean\right)\transpose $ is an $ \nTaxa \times \nTraits $ mean matrix and the across-taxa tree covariance matrix  $ \phylogenyVariance =  \diffusionVariance +  \rootpriorSamplesize \inverse \oneMatrix$ \citep{pybus2012}. The tree diffusion matrix $ \diffusionVariance $ is a deterministic function of $ \phylogeny $ and $ \oneMatrix $ is an $ \nTaxa \times \nTaxa $ matrix of all ones, such that the term $ \rootpriorSamplesize \inverse \oneMatrix $ comes from the integrated-out tree root prior. Figure \ref{fig:vmatrix} illustrates how the tree structure determines $\diffusionVariance$: the diagonals are equal to the sum of branch lengths from tip to root, and the off-diagonals are equal to the branch length from root to the most recent common ancestor of two tips.
\begin{figure}[h]
	\begin{center}
			\includegraphics[scale=0.35]{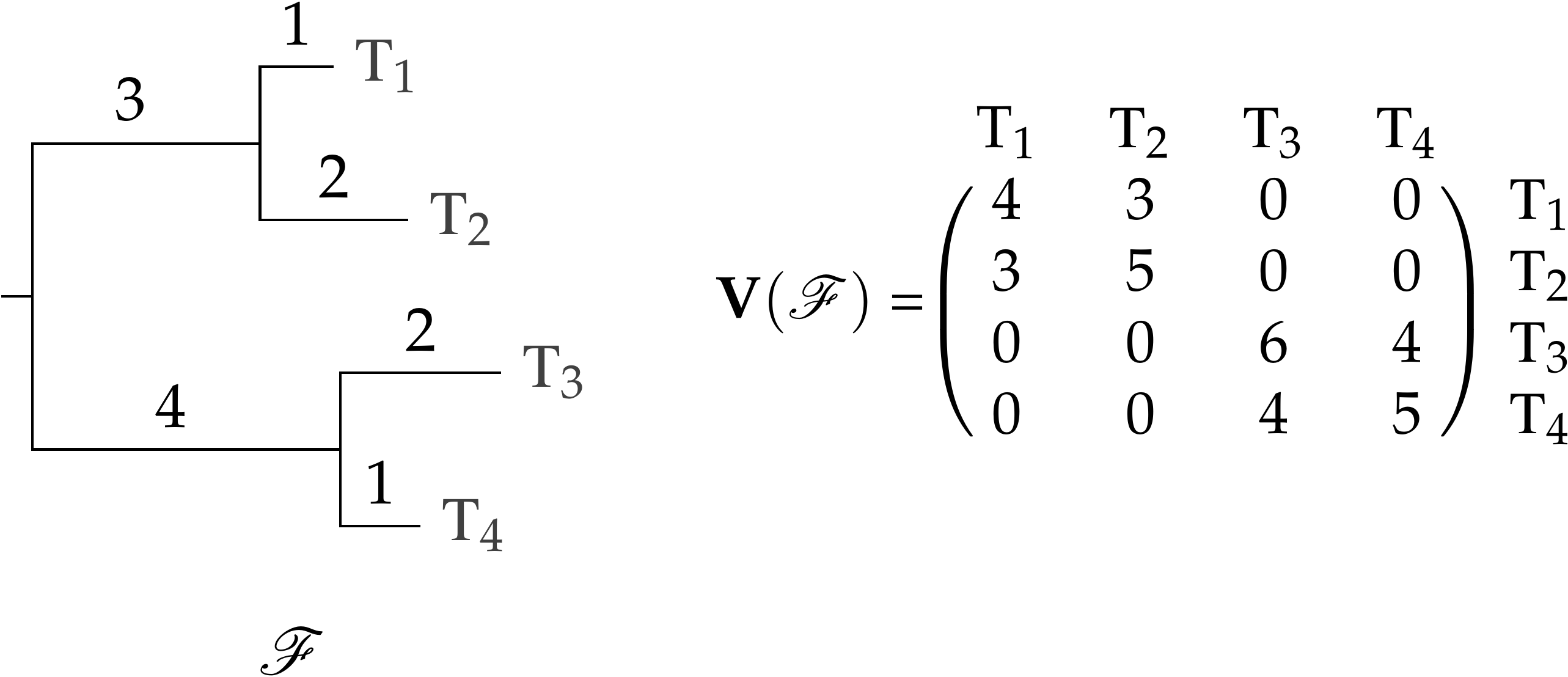}
	\end{center}
	\caption{A 4-taxon phylogenetic tree $\phylogeny$ with tips (T$_1$, T$_2$, T$_3$, T$_4$) and their corresponding tree diffusion matrix $ \diffusionVariance $.}
	\label{fig:vmatrix}
\end{figure}
\newcommand{\cIndicator}[2]{\ensuremath{\mathbb{I}(#1 \,|\,#2)}}
Combining equations \eqref{eq:thresholdFunc} and  \eqref{eq:matrixNormalDistri} enables us to write down the augmented likelihood of $ \latentData $ and $ \observedResponse $ through the factorization 
\begin{equation}\label{eq:3}
\cDensity{\observedResponse, \latentData}{\phylogenyVariance, \traitCovariance,\rootpriorMean, \rootpriorSamplesize, \mappingFunction}  = \cDensity{\observedResponse}{\latentData}  \cDensity{\latentData}{\phylogenyVariance,\traitCovariance,\rootpriorMean, \rootpriorSamplesize},
\end{equation}
where $ \cDensity{\observedResponse}{\latentData} = \cIndicator{\observedResponse}{\latentData, \mappingFunction} $, the indicator function that takes the value 1 if $ \latentData $ are consistent with the observations $ \observedResponse $ and 0 otherwise.
\newcommand{\traitDiagElement}[1]{D_{#1}}
\subsection{Decomposition of trait-covariance to account for varying data scales}\label{sec:DRD}
The previous work of \citet{Cybis2015} uses a conjugate Wishart prior on $\traitPrecision$ for computational convenience.
However, there are two problems with the Wishart prior.
First, with mixed-type data, it leaves the model not parameter-identifiable.
For a binary trait, we only know the sign of its latent parameter; the absolute value is arbitrary.
Consider a latent parameter $ \latentDataMatrixElement $ and its marginal trait variance $ \traitCovariance_{jj} $, the $ j $th diagonal element of $ \traitCovariance $.
If we scale them to $k \latentDataMatrixElement $ and $\traitCovariance_{jj}/k$ by any positive number $ k $, then according to (\ref{eq:matrixNormalDistri}), the likelihood remains unchanged.
Therefore, we need to fix the marginal variances for latent parameters underlying binary traits.
On the other hand, continuous traits can be seen as directly observed latent parameters, and their marginal trait variances depend on the potentially differing rates of change along $ \phylogeny $ and should be inferred from the data.
A Wishart prior on $\traitPrecision$ does not allow such distinct constraints on the marginal variances for binary and continuous traits.
The second problem with the Wishart prior is that strong dependencies exist among correlations and their joint distribution is considerably different from uniform \citep{tokuda2011visual}.
Without knowing the true correlation structure, these prior assumptions may not be appropriate.
Hence, we favor a noninformative, uniform prior on the correlation matrix.

We solve the above problems by decomposing $ \traitCovariance $ into an across-trait correlation matrix and standard deviations, with a jointly uniform prior on the correlation matrix. Specifically, we decompose $\traitCovariance= \traitDiag \traitCorr  \traitDiag$, where $ \traitCorr $ is the $ \nTraits \times \nTraits $ correlation matrix and $  \traitDiag $ is a diagonal matrix with elements $ \traitDiagElement{\nodeIndexOne\nodeIndexOne} = 1, \text{ for } i = 1, \dots, \nTraitsDiscrete$ and $ \traitDiagElement{\nodeIndexOne\nodeIndexOne} = \contTraitScale_\nodeIndexOne > 0 \text{ for } i = \nTraitsDiscrete + 1, \dots, \nTraits$.
We use the prior of Lewandowski, Kurowicka, and Joe (LKJ) on the positive-definite correlation matrix $ \traitCorr $ \citep{lkj2009},
with density
\begin{equation}
\mbox{LKJ}(\traitCorr|\eta) = \lkjdistribution{\traitCorr}{\eta}, 
\end{equation}
where $ \eta > 0$ is a shape parameter and $ c(\eta) $ is the normalizing constant. When $ \eta = 1 $, the LKJ prior implies a uniform distribution over all correlation matrices of dimension $ \nTraits $. For the diagonal standard deviation matrix $ \traitDiag $, we assume independent log normal priors on the variances $\contTraitScale_\nodeIndexOne ^2$ for $i = \nTraitsDiscrete + 1, \cdots, \nTraits$ with mean 0 and variance 1 on the log scale.
We describe how to carry out the posterior inference under this prior in Section~\ref{sec:hmc_for_trait_covariance}. There exists other methods for specifying a prior distribution on $ \traitDiag \traitCorr  \traitDiag$. 
For example, \citet{huang2013simple} use half-t distributions on standard deviations and achieve marginally uniform correlations. 
We prefer log normal priors over half-t because the latter has non-zero probability density for a zero standard deviation.
If one favors half-t standard deviations or marginally uniform correlations, our approach easily adapts to the prior in \citet{huang2013simple}.

\section{Inference}
Primary scientific interest lies in the across-trait correlation matrix $ \traitCorr$.
We integrate out the nuisance parameters by sampling from the joint posterior
\begin{equation}
\begin{aligned}
\cDensity{\traitCorr,\traitDiag, \latentData, \phylogeny}{\observedResponse,\sequenceData} &\propto \cDensity{\observedResponse}{\latentData} \:\times\: \cDensity{\latentData}{\traitCorr,\traitDiag, \phylogeny} \:\times\: \\& \quad  \quad  \density{\traitCorr, \traitDiag} \:\times\: \cDensity{\sequenceData}{\phylogeny} \:\times\: \density{\phylogeny}
\end{aligned}
\end{equation}
via a random-scan Gibbs scheme \citep{liu1995covariance}, and drop the posterior's dependence on the hyper-parameters $(\phylogenyVariance,\rootpriorMean, \rootpriorSamplesize, \mappingFunction)$ to ease notation. 
The joint posterior factorizes because sequences $ \sequenceData $ only affect the parameters of primary interest through $ \phylogeny $, since we assume $ \sequenceData $ to be conditionally independent of other parameters given $ \phylogeny $.

Within the Gibbs scheme, we alternatively update  $\latentData$, $(\traitCorr, \traitDiag)$ and $\phylogeny$ from their full conditionals, taking advantage of the conditional independence structure.
We construct  $ \cDensity{\sequenceData}{\phylogeny}$ from a continuous-time Markov chain evolutionary model \citep{suchard2001bayesian} that describes nucleotide substitutions along the branches of $\phylogeny$ that give rise to $\sequenceData$.
We assume a typical tree prior $\density{\phylogeny}$ based on a coalescent process \citep{kingman1982coalescent} and adopt a random-scan mixture of effective Metropolis-Hastings transition kernels \citep{beast2018} to update parameters that define $ \phylogeny $.
For more details on tree sampling and tree priors choices, we refer interested readers to \citet{beast2018}.
This section focuses on overcoming the scalability bottleneck of updating $\latentData$ from an $\nTaxa \nTraits$-dimensional truncated normal distribution by combining BPS with dynamic programming strategy.
We also describe how we deploy Hamiltonian Monte Carlo (HMC) to update $(\traitCorr, \traitDiag)$ to accommodate the non-conjugate prior on $\traitCovariance = \traitDiag \traitCorr  \traitDiag$.

\subsection{BPS for updating high-dimensional latent parameters}
BPS is a non-reversible ``rejection-free'' sampler originally introduced in the computational physics literature by \cite{peters2012} for simulating particle systems.
\cite{bouchard2018bouncy} later adopted
the algorithm with modifications to better suit statistical applications.
BPS explores a target distribution $ \bpsTarget(\bpsPosition) $ by simulating a piecewise deterministic Markov process.
The simulated particle follows a piecewise linear trajectory, with its evolution governed by the landscape of the \textit{energy} function $U(\bpsPosition) := - \log \bpsTarget(\bpsPosition)$.
To respect the target distribution, classical Monte Carlo algorithms first propose a move, then either accept or reject it such that a move towards areas of low probability
or, equivalently,
of high energy, is more likely to be rejected than one towards areas of high probability.
On the other hand, BPS modifies its particle trajectory via a Newtonian elastic collision against the energy gradient, thereby avoiding wasteful rejected moves.

BPS is an efficient sampler for log-concave target distributions in general, with the additional ability to account for parameter constraints by treating them as ``hard-walls'' against which the particle bounces. Of particular interest to us is the fact that, when the target distribution is a truncated MVN, the critical computation for BPS implementation 
is multiplying the precision matrix of the unconstrained MVN by an arbitrary vector. So BPS becomes an especially efficient approach when one can carry out these matrix-vector operations quickly. In our application, the tree diffusion process only defines the covariance, not the precision. But fortunately, the structured Brownian diffusion process enables us to efficiently compute the precision-vector products without costly matrix inversion.
BPS also allows us to condition on a subset of dimensions that correspond to the continuous traits without extra computation.
We begin with an overview of BPS following \citet{bouchard2018bouncy} and describe how to incorporate parameter constraints \citep{bierkens2018PDMP_with_constraint}; the subsequent sections describe how to optimize the implementation when sampling from a truncated MVN.
\subsubsection{BPS overview}
\label{sec:bps_overview}
To sample from the target distribution $ \bpsTarget(\bpsPosition) $, BPS simulates a particle with position $\bpsPosition(t)$ and velocity $\bpsVelocity(t)$ for time $t \geq 0$, initialized from $\bpsVelocity_{0} \sim \normal(\bm{0}, \I)$ and a given $\bpsPosition_0$ at time $t = 0$.
Over time intervals $t \in [t_k, t_{k+1}]$, the particle follows a piecewise linear path with velocity $ \bpsVelocity(t) = \bpsVelocity_{k} $ and position $\bpsPosition(t) = \bpsPosition_{k} + (t - t_k) \bpsVelocity_{k}$ .
An inhomogeneous Poisson process governs the inter-event times $\withinSegmentTime_{k + 1} = t_{k+1} - t_k$ with rate
\begin{equation}
\lambda(\bpsPosition(t), \bpsVelocity_{k})
= \max\left\{ 0, \left\langle \bpsVelocity_{k}, \nabla U(\bpsPosition(t)) \right\rangle \right\},
\end{equation}
where $ \left\langle\cdot, \cdot\right\rangle$ denotes an inner product.

When the target density is log-concave and differentiable, $U(\bpsPosition)$ is convex, so one can conveniently simulate the Markov process.
We describe how to simulate the process for a pre-specified amount of time $\totalTime > 0$, and the mapping $\bpsPosition_0 \to \bpsPosition(\totalTime)$ defines a Markov transition kernel with $\bpsTarget(\bpsPosition)$ as the stationary density:

\begin{enumerate}
	\item Solve a one-dimensional optimization problem to find
	\begin{equation}
	\label{eq:energy_minimization_prob}
	\withinSegmentTime_{\min} =  \argmin_{\withinSegmentTime \geq 0} U(\bpsPosition_{k - 1} + \withinSegmentTime \bpsVelocity_{k-1})
	\ \text{and} \
	U_{\min} = U(\bpsPosition_{k - 1} + \withinSegmentTime_{\min} \bpsVelocity_{k-1}).
	\end{equation}
	\item Draw $T \sim \text{Exp}(1)$, an exponential random variable with rate 1, and solve for the next inter-event time $ \withinSegmentTime_k $, the minimal root of
	\begin{equation}
	\label{eq:interevent_time}
	U(\bpsPosition_{k-1} + \withinSegmentTime_k \bpsVelocity_{k - 1}) - U_{\min} = T \text{ and } \withinSegmentTime_k > \withinSegmentTime_{\min}.
	\end{equation}
	\item Update $(\bpsPosition, \bpsVelocity)$ as
	\begin{equation}
	\label{eq:bounce_against_grad}
	\bpsPosition_{k} \gets \bpsPosition_{k - 1} + \withinSegmentTime_k \bpsVelocity_{k - 1}, \quad \bpsVelocity_{k} \gets \bpsVelocity_{k - 1} - 2 \frac{\left\langle \bpsVelocity_{k - 1}, \nabla U(\bpsPosition_{k}) \right\rangle}{\left\| \nabla U(\bpsPosition_{k}) \right\|^2 } \nabla U(\bpsPosition_{k}).
	\end{equation}
	\item Stop if $\sum_{j = 1}^{k} \withinSegmentTime_j \geq \totalTime $ and return $\bpsPosition(\totalTime) = \bpsPosition_{k - 1} + (\totalTime - t_{k - 1}) \bpsVelocity_{k - 1}$ where $t_{k - 1} = \sum_{j = 1}^{k - 1} \withinSegmentTime_j$, otherwise repeat Steps 1 - 3.
\end{enumerate}

Steps 1-4 form one conditional update by BPS inside a Gibbs scheme. They are the same as the basic BPS algorithm in \citet{bouchard2018bouncy}, except that we do not include velocity refreshment as random Poisson events.
Since we use BPS for conditional updates, we resample the velocity from $ \normal(\bm{0}, \I) $ at the beginning of every BPS iteration.
BPS without velocity refreshment is known to suffer from reducible behavior when applied to an isotropic multivariate normal distribution \citep{bouchard2018bouncy}.
Our velocity resampling already avoids this reducibility issue, and so we opt not to incorporate further refreshment inside the transition kernel. As long as the entire chain remains irreducible, Peskun-Tierney theory for non-reversible MCMC suggests that adding further events only reduces the efficiency \citep{bierkens2017limit, andrieu2019peskun}.

When the target distribution is constrained to some region $\bpsPosition \in \paramRegion$, the bounce events are caused not only by the gradient $\nabla U(\bpsPosition)$ but also by the domain boundary $\partial \paramRegion$.
We call these bounces ``gradient events" and ``boundary events" respectively.
Whichever occurs first is the actual bounce.
More precisely, we define the boundary event time $\withinSegmentTime_{\textrm{bd}, k} $ as
\begin{equation}
\label{eq:bdry_event_time}
\withinSegmentTime_{\textrm{bd}, k}
	= \inf_{\withinSegmentTime > 0} \left\{ \bpsPosition_{k - 1} + \withinSegmentTime \bpsVelocity_{k - 1} \notin \paramRegion
	\right\}.
\end{equation}
Then the bounce time is given by $\withinSegmentTime_k = \min\{\withinSegmentTime_{\textrm{bd}, k}, \withinSegmentTime_{\textrm{gr}, k}\}$, where $\withinSegmentTime_{\textrm{gr}, k}$ denotes the gradient event time of \eqref{eq:interevent_time}.
If $\withinSegmentTime_{\textrm{bd}, k} < \withinSegmentTime_{\textrm{gr}, k}$, we have a boundary bounce and the position is updated as in \eqref{eq:bounce_against_grad} while the velocity is updated as
\begin{equation}
\label{eq:bounce_against_bdry}
\bpsVelocity_{k}
	\gets \bpsVelocity_{k - 1}
	- 2 \left\langle \bpsVelocity_{k - 1}, \bm{\nu} \right\rangle \bm{\nu},
\end{equation}
where $\bm{\nu} = \bm{\nu}(\bpsPosition_k)$ is a unit vector orthogonal to the boundary at $\bpsPosition_k \in \partial \paramRegion$.

\subsubsection{BPS for truncated MVNs}
We now describe how the BPS simulation simplifies when the target density is a $d$-dimensional truncated MVN of the form
\begin{equation}
\label{eq:bps_targetDistribution_tMVN}
	\bpsPosition \sim \normalDistribution{\tipMeanMatrixVec}{\grandVariance}\text{ subject to } \bpsPosition \in \paramRegion = \{\sign(\bpsPosition) = \observedResponseVec \}
	\ \text{ for } \observedResponseVec \in \{\pm 1\}^d.
\end{equation}
Importantly, we can implement BPS so that, aside from basic and computationally inexpensive operations, it relies solely on matrix-vector multiplications by the precision matrix $\bpsPrecision = \grandPrecision$.
Moreover, under the orthant constraint $\{ \sign(\bpsPosition) = \observedResponseVec \}$, we can handle a bounce against the boundary in a particularly efficient manner, only requiring access to a column of $\bpsPrecision$.

We start with gradient events and then describe how to find boundary event times. Now $ U(\bpsPosition) = - \log p(\bpsPosition) = \frac{1}{2} (\bpsPosition - \tipMeanMatrixVec)^\intercal \bpsPrecision (\bpsPosition - \tipMeanMatrixVec) + \constant$ where constant $ \constant $ does not depend on $ \bpsPosition $, therefore
\begin{multline}
\label{eq:truncated_normal_energy}
U(\bpsPosition + \withinSegmentTime \bpsVelocity)
	= \frac{1}{2} \langle \bpsVelocity, \bpsPhiv \rangle \withinSegmentTime^2
	+ \langle \bpsVelocity, \bpsPhix \rangle \withinSegmentTime
	+ \frac{1}{2} \langle \bpsPosition - \tipMeanMatrixVec, \bpsPhix \rangle + \constant \\
	\ \text{ where } \ \bpsPhiv = \bpsPrecision \bpsVelocity
	\ \text{ and } \ \bpsPhix = \bpsPrecision (\bpsPosition - \tipMeanMatrixVec) = \nabla U(\bpsPosition).
\end{multline}
The solution to the optimization problem \eqref{eq:energy_minimization_prob} is given by
\begin{equation}
\label{eq:truncated_normal_minimizer}
\begin{aligned}
\withinSegmentTime_{\min}
	&= \max\left\{0,  - \langle \bpsVelocity, \bpsPhix \rangle \big/ \langle \bpsVelocity, \bpsPhiv \rangle \right\}, \\
U_{\min}
	&= \frac{1}{2} \langle \bpsVelocity, \bpsPhiv \rangle \withinSegmentTime_{\min}^2
		+ \langle \bpsVelocity, \bpsPhix \rangle \withinSegmentTime _{\min}
		+ \frac{1}{2} \langle \bpsPosition - \tipMeanMatrixVec, \bpsPhix \rangle + \constant.
\end{aligned}
\end{equation}
It follows from \eqref{eq:truncated_normal_energy} that the gradient event time in \eqref{eq:interevent_time} coincides with the larger root of the quadratic equation  $a \withinSegmentTime^2 + b \withinSegmentTime + c = 0$ with
\begin{gather*}
a =  \frac{1}{2} \langle \bpsVelocity, \bpsPhiv \rangle, \
	b = \langle \bpsVelocity, \bpsPhix \rangle,
	\text{ and } \,
c
	= - \frac{1}{2} \langle \bpsVelocity, \bpsPhiv \rangle \withinSegmentTime_{\min}^2
		- \langle \bpsVelocity, \bpsPhix \rangle \withinSegmentTime _{\min} - T,
\end{gather*}
so
\begin{equation*}
\withinSegmentTime_{\textrm{gr}}
	= \frac{-b + \sqrt{b^2 - 4 a c}}{2a}.
\end{equation*}
When a gradient event takes place, the position and velocity are updated according to \eqref{eq:bounce_against_grad} with
\begin{equation}
\label{eq:gradient_update_for_truncated_normal}
 \nabla U(\bpsPosition + \withinSegmentTime \bpsVelocity)
	= \bphi_{\bpsPosition + \withinSegmentTime \bpsVelocity}
	= \bpsPrecision (\bpsPosition - \tipMeanMatrixVec) + \withinSegmentTime \bpsPrecision \bpsVelocity
	= \bpsPhix + \withinSegmentTime \bpsPhiv.
\end{equation}
Note that $\bphi_{\bpsPosition + \withinSegmentTime \bpsVelocity}$ can be computed by an element-wise addition of $\bpsPhix$ and $\withinSegmentTime \bpsPhiv$, rather than the expensive matrix-vector operation $\bpsPosition + \withinSegmentTime \bpsVelocity \to \bpsPrecision (\bpsPosition + \withinSegmentTime \bpsVelocity )$.

The orthant boundary is given by $\cup_i \{x_i = 0\}$.
When $\sign(x_i) = \sign(v_i)$, where $x_i$ and $v_i$ denotes the $ i $-th coordinate of particle position and velocity, the particle is moving away from the $i$-th coordinate boundary $\{x_i = 0\}$ and thus never reaches it.
Otherwise, the coordinate boundary is reached at time $\withinSegmentTime = | x_i / v_i |$.
Hence $\withinSegmentTime_{\textrm{bd}}$ can be expressed as
\begin{equation*}
\withinSegmentTime_{\textrm{bd}}
	= \left| x_{i_\textrm{bd}} / v_{i_\textrm{bd}} \right|, \
{i_\textrm{bd}} = \textstyle \argmin_{\, i \in I} \left| x_i / v_i \right|
	\, \text{ for } \,
	I = \{i : x_i v_i < 0 \}.
\end{equation*}
When a boundary event takes place, the particle bounces against the plane orthogonal to the standard basis vector $\bm{\nu} = \basisVector{i_\textrm{bd}}$.
As the updated velocity takes the form $\bpsVelocity^* \gets \bpsVelocity - 2 v_{i_\textrm{bd}} \bm{e}_{i_\textrm{bd}}$, we can save computational cost of simulating the next line segment by realizing that
\begin{equation}
\label{eq:efficient_precision_velocity_update}
\bphi_{\bpsVelocity^*} = \bpsPrecision \bpsVelocity^*
	= \bpsPhiv + 2 v^*_{i_\textrm{bd}} \bpsPrecision \bm{e}_{i_\textrm{bd}}
	\ \text{ where } \  v^*_{i_\textrm{bd}} = - v_{i_\textrm{bd}}.
\end{equation}
In other words, we can compute $\bphi_{\bpsVelocity^*}$ by simply extracting the $i_\textrm{bd}$-th column of $\bpsPrecision$ and updating $\bpsPhiv$ with an element-wise addition.
This avoids the expensive matrix-vector operation $\bpsVelocity^* \to \bpsPrecision \bpsVelocity^*$.

Algorithm~\ref{alg:bps_conditional} describes BPS implementation for truncated MVNs based on the discussion above, with the most critical calculations optimized.
Within each line segment, $\bpsPhix$ and $\bpsPhiv$ once efficiently computed (Section \ref{sec:dynamicProgramming}) can be re-used throughout. In our application the observed continuous traits correspond to fixed dimensions in $ \bpsPosition $, so we slightly modify the BPS such that it can sample from a conditional truncated MVN.
Specifically, we partition $ \bpsPosition = \left( \bpsPosition_b , \bpsPosition_c \right)$
by latent ($ \bpsPosition_b $) and observed dimensions ($ \bpsPosition_c  $), with the aim to generate samples from the conditional distribution $ \cDensity{\bpsPosition_b}{\bpsPosition_c}$  (details in Appendix \ref{sec:conditional tMVN}). We choose the tuning parameter $ \totalTime $ based on a heuristic that works well in practice (Section \ref{sec:tuningBPS}).
\algnewcommand{\LeftComment}[1]{\Statex \(\triangleright\) #1}

\begin{algorithm}[h!]
	\caption{Bouncy particle sampler for multivariate truncated normal distributions}
	\label{alg:bps_conditional}
	\begin{algorithmic}[1]
		\Require $\totalTime, \text{initial value for } \bpsPosition$
		\State $\bpsVelocity \sim \normal(\bm{0}, \I)$
		\State $\bm{\varphi}_{\bpsPosition} \gets \bpsPrecision (\bpsPosition - \tipMeanMatrixVec) $
		\Comment{$\bm{\varphi}_{\bpsPosition} = \nabla U(\bpsPosition)$ is the gradient of energy}
		\While{$\totalTime > 0$}
		\vspace{2mm}
		\LeftComment{compute reused quantities once}
		\If{previous bounce is a boundary event at coordinate $ i $}
			\State $\bm{\varphi}_{\bpsVelocity} \gets \bm{\varphi}_{\bpsVelocity} +  2 v_{i} \bpsPrecision \basisVector{i}$
		\Else
			\State $\bm{\varphi}_{\bpsVelocity} \gets \bpsPrecision \bpsVelocity $ \Comment{the expensive step}
		\EndIf
		\State $\varphi_{\bpsVelocity, \bpsPosition} \gets \bpsVelocity^\intercal \bm{\varphi}_{\bpsPosition},
		\varphi_{\bpsVelocity, \bpsVelocity} \gets \bpsVelocity^\intercal \bm{\varphi}_{\bpsVelocity}$
		\vspace{2mm}
		\LeftComment{find gradient event time}
		\State $\withinSegmentTime_{\rm min} \gets \max\left\{ 0, - \varphi_{\bpsVelocity, \bpsPosition} / \varphi_{\bpsVelocity, \bpsVelocity} \right\}$
		\State $T \sim \text{Exp}(1)$
		\State $
			a \gets \frac{1}{2} \varphi_{\bpsVelocity, \bpsVelocity}, \
			b \gets \varphi_{\bpsVelocity, \bpsPosition}, \
			c \gets - \frac{1}{2} \withinSegmentTime_{\min}^2 \varphi_{\bpsVelocity, \bpsVelocity} - \withinSegmentTime_{\min} \varphi_{\bpsVelocity, \bpsPosition} - T
		$
		\State $\gradientTime \gets (- b + \sqrt{b^2 - 4 a c}) / (2 a)$
		
		\vspace{2mm}
		\LeftComment{find truncation event time at coordinate $ i $}
		\State $\reflectTime \gets \argmin_i x_{\nodeIndexOne} /  v_\nodeIndexOne, \text{ for } \nodeIndexOne \text{ with }x_{\nodeIndexOne}  v_\nodeIndexOne < 0 $
		
		\vspace{2mm}
		\LeftComment{bounce happens}
		\State $\bounceTime  \gets \min \left\{ \gradientTime, \reflectTime, \totalTime \right\} $
		\State $\bpsPosition \gets \bpsPosition + \bounceTime \bm{v}$,
		$\bm{\varphi}_{\bpsPosition} \gets \bm{\varphi}_{\bpsPosition} + \bounceTime \bm{\varphi}_{\bpsVelocity}$
		\If{$  \bounceTime = \reflectTime$}
		\State $v_i \gets -v_i$
		\ElsIf{$ \bounceTime = \gradientTime$}
			\State $\bpsVelocity \gets \bpsVelocity - (2 \left\langle \bpsVelocity, \bm{\varphi}_{\bpsPosition} \right\rangle \big/ \| \bm{\varphi}_{\bpsPosition} \|^2) \, \bm{\varphi}_{\bpsPosition}$
		\EndIf
				\State $\totalTime \gets \totalTime - \bounceTime$
		\EndWhile
	\end{algorithmic}
\end{algorithm}

\subsubsection{Dynamic programming strategy to overcome computational bottleneck} 
\label{sec:dynamicProgramming}
A straight implementation of BPS remains computationally challenging, as computing $\bpsPhix$ and $\bpsPhiv$ in Algorithm \ref{alg:bps_conditional} involves a high-dimensional matrix inverse when the model is parameterized in terms of $\grandVariance$.
From \eqref{eq:matrixNormalDistri} and the equivalence between matrix normal  and multivariate normal distributions, to sample latent parameters $ \latentData $ from their conditional posterior, the target distribution \eqref{eq:bps_targetDistribution_tMVN} specifies as $ \latentDataVec = \vectorize{\latentData}$, $ \tipMeanMatrixVec = \vectorize{\tipMeanMatrix}$, $ \grandVariance = \grandVarianceExpression $, and $ \observedResponseVec = \vectorize{\observedResponse} $, where $ \vectorize{\cdot} $ is the vectorization that converts an $ \nTaxa \times \nTraits $ matrix into an $ \nTaxa\nTraits \times 1 $ vector and $ \kronecker$ denotes the Kronecker product.
A naive matrix inverse operation $\grandPrecision = \grandPrecisionExpression$ has an intimidating complexity of $ \order{\nTaxa^3 + \nTraits^3} $.
If we have a fixed tree, such that $\phylogenyPrecision$ is known, the typical computation proceeds via
\begin{equation} \begin{aligned}\label{eq:kroneckerEq}
\grandVariance\inverse \left(\latentDataVec - \tipMeanMatrixVec \right)  =
\left( \grandPrecisionExpression \right)
\left(\latentDataVec - \tipMeanMatrixVec \right)
&= \vectorize{ \phylogenyPrecision \left(\latentData - \tipMeanMatrix\right) \traitPrecision },
\end{aligned} \end{equation}
with a cost ${\cal O}\hspace{-0.1em}\left( \nTaxa^2 \nTraits +\nTaxa \nTraits^2\right)$.
When the tree is random, the $ \order{\nTaxa^3} $ cost to get $ \phylogenyPrecision $ seems unavoidable. However, we show that even with a random tree, evaluating  $\bpsPhix$ and $\bpsPhiv$ can be $\order{ \nTaxa \nTraits^2}$.
We use conditional densities to evaluate these products (Proposition \ref{proposition1}) and obtain all conditional densities simultaneously via a dynamic programming strategy that avoids explicitly inverting $ \phylogenyVariance $.
\begin{prop} \label{proposition1}
	Given joint variance matrix $\grandVariance$ and vectorized latent data $\latentDataVec$, the energy gradient $ \nabla U(\bpsPosition) $ is
	\begin{equation} \label{linear}
	\bpsPhix = \grandVariance\inverse \left(\latentDataVec  - \tipMeanMatrixVec \right)  =
	\left(
	\begin{array}{c}
	\nodePrePrecision{1}\left( \latentData_1 - \nodePreMean{1} \right) \\
	\vdots \\
	\nodePrePrecision{\nTaxa}\left( \latentData_{\nTaxa} - \nodePreMean{\nTaxa}
	\right)
	\end{array} \right),
	\end{equation}
	where $\nodePreMean{\nodeIndexOne}$ and $\nodePrePrecision{\nodeIndexOne}$ are the mean  and the precision matrix of the distributions  $\cDensity{\latentData_{\nodeIndexOne}}{\latentData_{\excludeNode{\nodeIndexOne}}}$ for $\nodeIndexOne = 1,\ldots,\nTaxa$, and $\cDensity{\latentData_{\nodeIndexOne}}{\latentData_{\excludeNode{\nodeIndexOne}}}$ is the conditional distribution of latent parameters at one tree tip given those of all the other tips.
	\label{conditionalFactorization}
\end{prop}
\begin{proof}
	$\latentDataVec \sim \normalDistribution{\tipMeanMatrixVec}{\grandVariance}$, so $\cDensity{\latentData_{\nodeIndexOne}}{\latentData_{\excludeNode{\nodeIndexOne}}}$ are also multivariate normal. 
	Note that
	\begin{equation}\begin{aligned}
	\gradient{\latentDataVec} \left[ \log \density{\latentDataVec} \right]
	&= - \frac{1}{2} \grandVariance\inverse \left(\latentDataVec - \tipMeanMatrixVec \right).
	\label{lhsGradient}
	\end{aligned}\end{equation}
	Likewise, $
	\gradient{\latentDataVec} \left[ \log \density{\latentDataVec} \right] =
	\left(
	\gradient{\latentData_{1}} \left[ \log \density{\latentDataVec} \right], \ldots, \gradient{\latentData_{N}} \left[ \log \density{\latentDataVec} \right]
	\right)\transpose
	$ with
	\begin{equation}\begin{aligned}
	\gradient{\latentData_{\nodeIndexOne}} \left[ \log \density{\latentDataVec} \right] &=
	\gradient{\latentData_{\nodeIndexOne}} \left[
	\log \cDensity{\latentData_{\nodeIndexOne}}{\latentData_{\excludeNode{\nodeIndexOne}}} +
	\log \density{\latentData_{\excludeNode{\nodeIndexOne}}}
	\right] \\
	&= \gradient{\latentData_{\nodeIndexOne}} \left[
	\log \cDensity{\latentData_{\nodeIndexOne}}{\latentData_{\excludeNode{\nodeIndexOne}}} \right] \\
	&= - \frac{1}{2} \nodePrePrecision{\nodeIndexOne}\left( \latentData_{\nodeIndexOne} - \nodePreMean{\nodeIndexOne} \right).
	\label{rhsGradient}
	\end{aligned}\end{equation}
	Equating \eqref{lhsGradient} and \eqref{rhsGradient} completes the proof.
\end{proof}
In Proposition \ref{proposition1}, the partition is by taxon, but we can generalize to any arbitrary partitioning of the dimensions.
By replacing $ \latentDataVec  - \tipMeanMatrixVec $ with $ \bpsVelocity $ (or $ \basisVector{i} $), we achieve a similar result for $ \bpsPhiv$ (or $ \bpsPrecision\basisVector{i}$).
Given $\nodePreMean{\nodeIndexOne}$ and $\nodePrePrecision{\nodeIndexOne}$,
the $\order{ \nTaxa \nTraits^2}$ matrix-vector operation $\bpsVelocity^* \to \bpsPrecision \bpsVelocity^*$ based on Proposition \ref{proposition1} is generally required for updating $\bphi_{ \bpsVelocity^*}$, but for boundary bounces, we can exploit \eqref{eq:efficient_precision_velocity_update} and update
$\bphi_{ \bpsVelocity^*}$ in $\order{ \nTaxa \nTraits}$. For the conditional posterior distribution in our HIV application (Section~\ref{sec:application}), boundary bounces occur far more frequently than gradient ones and thus the efficient update via \eqref{eq:efficient_precision_velocity_update} leads to further significant speed-up.

Fortunately, we are able to efficiently compute $\nodePreMean{\nodeIndexOne}$ and $\nodePrePrecision{\nodeIndexOne}$ through a dynamic programming strategy that recursively traverses the tree \citep{pybus2012} and enjoys a complexity of $\order{\nTaxa \nTraits}$. Here we give the results and omit the derivatives found in \citet{pybus2012} and \citet{Cybis2015}.

The recursive traversals visit every node first in post-order (child $\rightarrow$ parent) and then again in pre-order (parent $\rightarrow$ child) to calculate partial data likelihoods that lead to $\nodePreMean{\nodeIndexOne}$ and $\nodePrePrecision{\nodeIndexOne}$. The post-order traversal begins at a tip and ends at the root, while pre-order starts at the root and reaches every tip. The following results are in terms of the node triplets $(\treeChildOne, \treeChildTwo,\treeParent)$ where $\text{pa}(\treeChildOne) = \text{pa}(\treeChildTwo) = \treeParent$ as in Figure \ref{fig:traversalTree}. We define $\below{\treeChildOne}$ as the tree tips that are descendants to or include (``below'') node $ \treeChildOne$ and $\above{\treeChildOne}$ as the tree tips that are not descendants to (``above'') node $ \treeChildOne $.

\begin{figure}[h]
	\centering
	\includegraphics[scale=0.6]{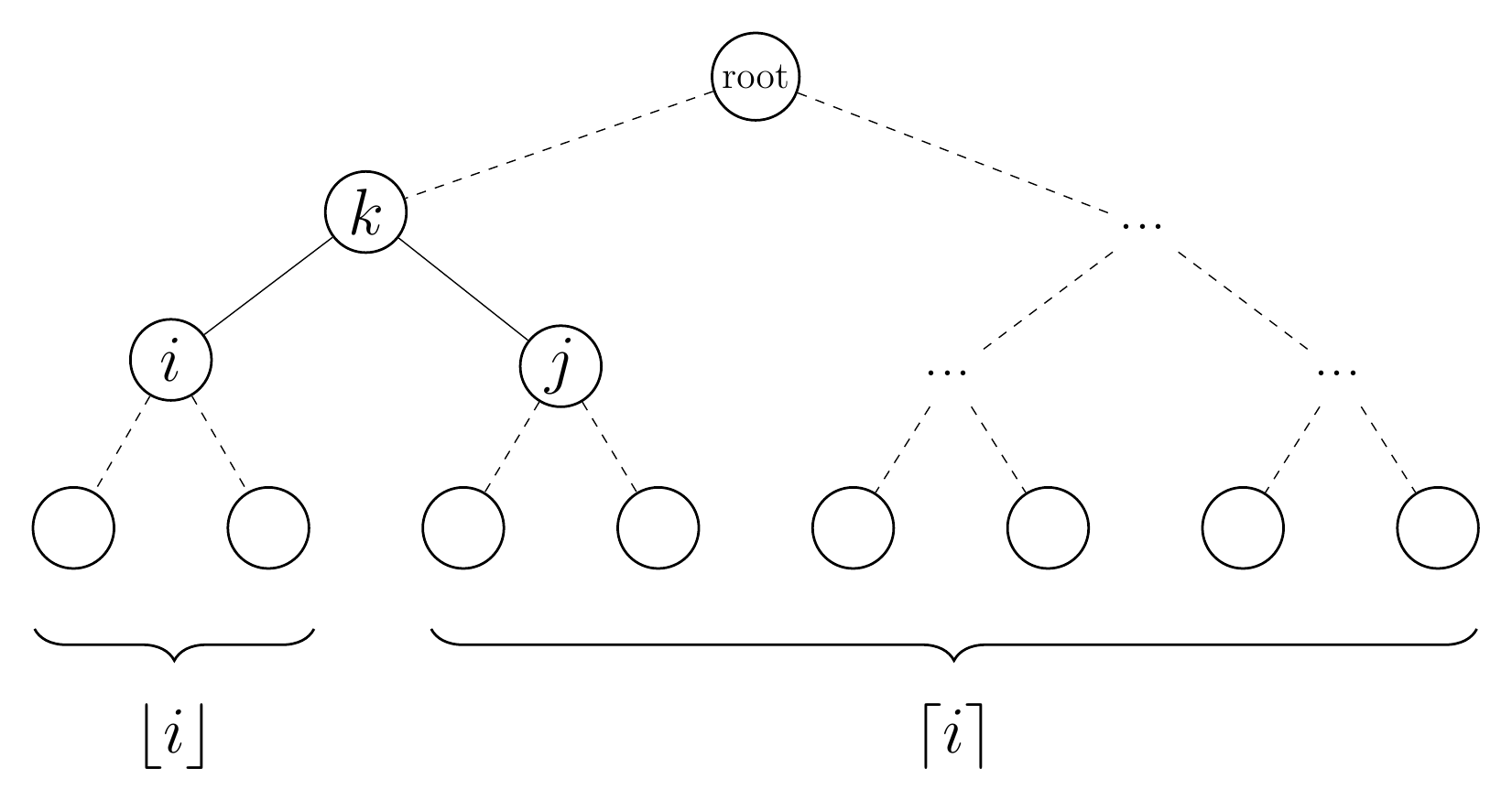}
	\caption{A sample tree to illustrate post- and pre- traversals for efficiently computing $\cDensity{\latentData_{\nodeIndexOne}}{\latentData_{\excludeNode{\nodeIndexOne}}}$. In the triplet $(\treeChildOne, \treeChildTwo, \treeParent)$, parent node $\treeParent$ has two children  $\treeChildOne$ and $\treeChildTwo$.
	We group the tip nodes into two disjoint and exhaustive classes: $\below{\treeChildOne} = $ tree tips that are descendants to or include node $ \treeChildOne$ and $\above{\treeChildOne} = $ tree tips that are not descendants to $ \treeChildOne $. }
	\label{fig:traversalTree}
\end{figure}

During the post-order traversal, the partial likelihoods of the data $\latentData_{\below{\treeChildOne}}$ given latent $\latentData_{\treeChildOne}$ is proportional to a MVN density of $\latentData_\treeChildOne$, in terms of a post-order mean $ \nodeMean{\treeChildOne} $ and variance $ \postOrderScalar{\treeChildOne}\traitCovariance $ \citep{pybus2012}, that is,
\begin{equation}
\cDensity{\latentData_{\below{\treeChildOne}}}{\latentData_\treeChildOne} \propto \mvn{\latentData_\treeChildOne}{\nodeMean{\treeChildOne}}{\postOrderScalar{\treeChildOne}\traitCovariance}.
\end{equation}
We re-employ these quantities shortly in the pre-order traversal.  At the tree tips, $ \nodeMean{\treeChildOne} = \latentData_\treeChildOne$ and the variance scalar $ \postOrderScalar{\treeChildOne} = 0$. For internal nodes,
\begin{equation}\label{eq:postorderScalar}
\begin{aligned}
\nodeMean{\treeParent} & = \postOrderScalar{\treeParent}\left[\left(\postOrderScalar{\treeChildOne} + \branchLength{\treeChildOne}\right)\inverse \nodeMean{\treeChildOne} +  \left(\postOrderScalar{\treeChildTwo} + \branchLength{\treeChildTwo}\right)\inverse\nodeMean{\treeChildTwo}\right], \text{ with} \\
\postOrderScalar{\treeParent} & =
\left[\left(\postOrderScalar{\treeChildOne} + \branchLength{\treeChildOne}\right)\inverse +  \left(\postOrderScalar{\treeChildTwo} + \branchLength{\treeChildTwo}\right)\inverse\right]\inverse.
\end{aligned}
\end{equation}
Similarly, for the pre-order traversal, we calculate the conditional density of $\latentData_{\treeChildOne}$ at node $ \nodeIndexOne $ given the data above it,
\begin{equation}
\cDensity{\latentData_{\treeChildOne}}{\latentData_{\above{\treeChildOne}}} \propto \mvn{\latentData_{\treeChildOne}}{\nodePreMean{\treeChildOne}}{\preOrderScalar{\treeChildOne}\traitCovariance},
\end{equation}
in terms of a pre-order mean $ \nodePreMean{\treeChildOne} $ and variance $ \preOrderScalar{\treeChildOne}\traitCovariance $. Starting from the root where $ \preOrderScalar{2\nTaxa - 1}= \rootpriorSamplesize\inverse $ and $ \nodePreMean{2\nTaxa - 1} =  \rootpriorMean$, the traversal proceeds via
\begin{equation}\label{eq:preorderScalar}
\begin{aligned}
\nodePreMean{\treeChildOne} &=  \preOrderScalar{\treeChildOne}^*\left[\left(\postOrderScalar{\treeChildTwo} + \branchLength{\treeChildTwo}\right)\inverse \nodeMean{\treeChildTwo} +  \preOrderScalar{\treeParent}\inverse\nodePreMean{\treeParent}\right], \text{ with} \\
\preOrderScalar{\treeChildOne}^* &=
\left[\left(\postOrderScalar{\treeChildTwo} + \branchLength{\treeChildTwo}\right)\inverse +  \preOrderScalar{\treeParent}\inverse\right] \inverse, \text{ and } \\
\preOrderScalar{\treeChildOne} &=
\preOrderScalar{\treeChildOne} ^*+ \branchLength{\treeChildOne}.
\end{aligned}
\end{equation}
When reaching the tips where $ \above{\nodeIndexOne} =  \protect\excludeNode{\nodeIndexOne}$, we obtain both the desired conditional mean $ \nodePreMean{\nodeIndexOne}$ and precision $ \nodePrePrecision{\nodeIndexOne} = \left( \preOrderScalar{\treeChildOne} \traitCovariance \right)\inverse$.

For both pre- and post-order traversals, at each node we require $\order{\nTraits}$ elementary operations
to obtain the mean vector and variance scalar; so, visiting all the nodes costs $ \order{\nTaxa \nTraits} $. With $\nodePreMean{\nodeIndexOne}$ and $\nodePrePrecision{\nodeIndexOne} $ for $\nodeIndexOne = 1, \ldots, \nTaxa$ ready in hand, the computation in \eqref{linear} remains $ \order{\nTaxa \nTraits^2} $.
\subsection{Hamiltonian Monte Carlo for updating trait covariance components}
\label{sec:hmc_for_trait_covariance}
The across-trait covariance components $ \traitCorr $ and $ \traitDiag $ have complex and high-dimensional full conditional distributions, with no obvious structure to admit sampling via specialized algorithms.
We therefore rely on HMC \citep{hmcneal}, a state-of-the-art general purpose sampler.
HMC only requires evaluations of the log-density and its gradient, yet is capable of sampling efficiently from complex high-dimensional distributions \citep{bda13}.

To introduce the main ideas behind HMC, we denote the distribution of interest by $p(\btheta) = \cDensity{\traitCorr, \traitDiag}{\latentData, \phylogeny}$.
In order to sample from $\btheta = \left( \traitCorr, \traitDiag \right)$, HMC introduces an auxiliary \textit{momentum} variable $\momentum \sim \mathcal{N}(\bm{0}, \M)$ and samples from the product density $p(\btheta, \momentum) = p(\btheta) p(\momentum)$.
HMC explores the joint space $(\btheta, \momentum)$ by approximating Hamiltonian dynamics that evolve according to the differential equation:
\begin{equation}
\label{eq:hamiltons_equation}
\frac{{\rm d} \btheta}{{\rm d} t}
=  \momentum, \quad
\frac{{\rm d} \momentum}{{\rm d} t}
= \nabla \log p(\btheta).
\end{equation}
More precisely, each HMC iteration proceeds as follows. We first draw a new value of $\momentum$ from its marginal distribution, then we approximate the evolution in \eqref{eq:hamiltons_equation} from time $t = 0$ to $t = \tau$ by applying $L = \lfloor \tau / \stepsize \rfloor$ steps of the \textit{leapfrog} update with stepsize $\epsilon$:
\begin{equation}
\momentum  \gets \momentum + \frac{\stepsize}{2} \nabla_{\btheta} \log p(\btheta), \quad
\btheta	\gets \btheta + \stepsize  \momentum, \quad
\momentum \gets \momentum + \frac{\stepsize}{2} \nabla_{\btheta} \log p(\btheta).
\end{equation}
The end point of the approximated dynamics constitutes a valid \textit{Metropolis} proposal \citep{metropolis53} that is accepted or rejected according to the standard acceptance probability formula.

By virtue of the properties of Hamiltonian dynamics, the HMC proposals generated above can be far away from the current state yet be accepted with high probability.
Good performance of HMC depends critically on well-calibrated choices of $ L $ and $ \epsilon $.
We automate these choices via the stochastic optimization approach of \cite{andrieu08} and the \textit{No-U-Turn} algorithm of \cite{hoffman2014nuts} that have been shown to achieve performance competitive with manually optimized HMC.
Because HMC applies most conveniently to a distribution without parameter constraints, we map $ \traitCorr $ and $ \traitDiag $ to an unconstrained space using standard transformations \citep{stan18}.

\section{Application on HIV immune escape}
\label{sec:application}
\subsection{Background}\label{sec:application_background}
As a rapidly evolving RNA virus, HIV-1 has established extensive genetic diversity that researchers classify into different major groups and, for HIV-1 group M, into different subtypes \citep{Hemelaar:2012lz}.
Such diversity implies that phenotypic traits can vary remarkably among strains circulating in different patients. Differences in viral virulence and their determinants, together with host factors, may explain the large variability in disease progression rates among patients. On the host side, human leukocyte antigen (HLA) class I alleles are important determinants of immune control that are known to be associated with differential HIV disease outcomes, with particular HLA alleles offering considerable protective effect \citep{goulder2012hiv}. An interesting phenomenon is that HIV-1 can evolve to escape the HLA-mediated immune response, but the responsible escape mutations may compromise fitness and hence reduce viral virulence \citep{nomura2013significant,Payne2014}.
Identifying these mutations and their effect on virulence while controlling for the evolutionary relationships among the viruses that spread in populations with heterogeneous HLA backgrounds represents a particular challenge. Here, we address this by estimating the posterior distribution of across-trait correlation while controlling for the unknown viral evolutionary history.

We analyze a data set of $ \nTaxa = 535$ aligned HIV-1 \textit{gag} gene sequences collected from 535 patients in Botswana and South Africa between 2003 and 2010 \citep{Payne2014}.
Both countries are severely affected by the subtype C variant of HIV-1 group M.
Each sequence is associated with a known sampling date and phenotypic measurements, including $ \nTraitsCont = 3 $ continuous traits that are replicative capacity (RC), viral load (VL), and cluster of differentiation 4 (CD4) cell count. An increasing VL and a decreasing CD4 count in the asymptomatic stage characterize a typical HIV infection; RC is a viral fitness measure obtained by an assay that, in this case, assesses the growth rate of recombinant viruses containing the patient-specific \textit{gag-protease} gene relative to a control virus \citep{Payne2014}.
We further link each sequence with $ \nTraitsDiscrete = 21 $ binary traits, including the presence/absence of candidate HLA-associated escape mutations at 20 different amino acid positions in the \textit{gag} protein, and another binary trait for the country of  sampling (Botswana or South Africa).
In cases where ambiguous nucleotide states in a codon prevent the determination of presence/absence of escape mutations, we encode binary trait states as unobserved (ranging from 0.2\% to 21\% across taxa) and set them as unbounded dimensions in the truncated normal distribution sampled by BPS.
\subsection{Correlation among traits}\label{sec:application_traitCorr}
We revisit the original study questions in \citet{Payne2014} concerning the extent to which HLA-driven HIV adaptation impacts virulence in both Botswana and South Africa populations. Differences in HIV adaptation and virulence may arise from the fact the HIV epidemic in Botswana precedes that in South Africa, leaving more time for the virus to adapt to protective HLA alleles. Our approach employing a Bayesian inference framework based on the phylogenetic multivariate probit model, is substantially different from \citet{Payne2014} as they did not control for the shared evolutionary history between samples. For this $ \nTaxa = 535, \nTraitsDiscrete = 21 , \nTraitsCont = 3 $ data set,  after fitting the phylogenetic multivariate probit model, we obtain posterior samples for parameters that are of scientific interest.
For MCMC convergence assessment, we run the chain until the minimal effective sample size (ESS) across all dimensions of $\latentData$, $\traitCorr$ and $\traitDiag$ is above 200.
This takes about $ 10^7 $ individual transition kernel applications under our random-scan Gibbs scheme (iterations) and 30 hours on an Amazon EC2 c5.large instance, and we discard the first 10\% of the samples as burn-in.  As a further diagnostic, we execute five independent chains and confirm that the potential scale reduction statistic $ \hat{R} $ for all correlation elements fall within range [1, 1.04], well below the standard convergence criterion of 1.1 \citep{gelman1992inference}.
We implement the method in the software BEAST \citep{beast2018}, and provide the data set and source code in the Supplementary Material \citep{supple}.

The heat map in Figure \ref{fig:3} depicts significant across-trait correlation determined by a 90\% highest posterior density (HPD) interval that does not contain zero. We mainly focus on the last 4 rows that relate to questions addressed by \citet{Payne2014}, e.g.~difference in HLA escape mutations between the two countries and correlation between escape mutations and infection traits (VL and CD4 count) as well as replicative capacity.
\begin{figure}[h!]
	\centering
	\includegraphics[scale=0.45]{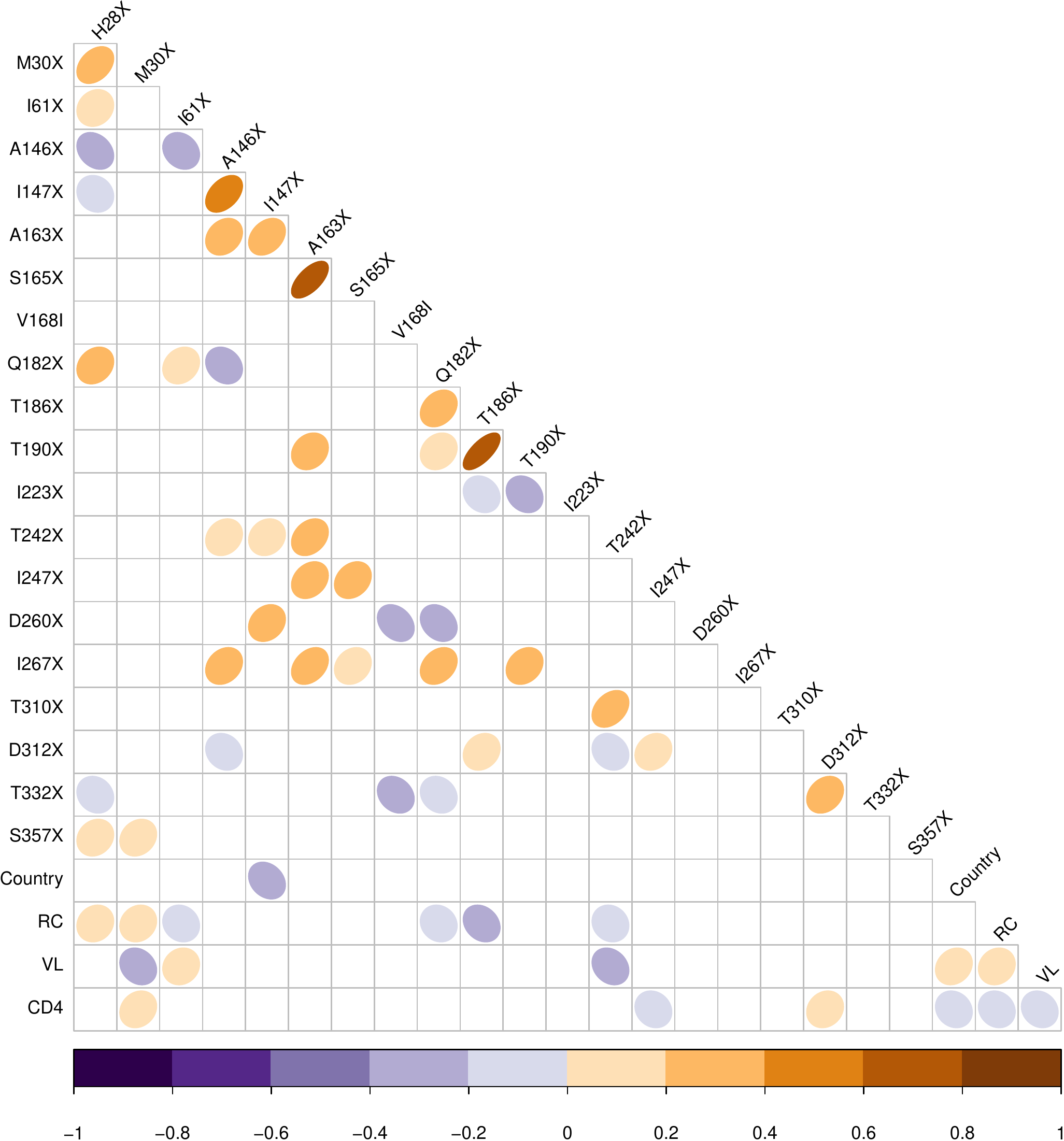}
	\caption{Significant across-trait correlation with $ < 10\%$ posterior tail probability and their posterior mean estimates (in color). HIV \textit{gag} mutations are named by the wild type amino acid state, the amino acid site number according to the standard reference genome (HXB2), and the amino acid `escape' state that is any other amino acid or a deletion (`X') in almost all cases.
		Country = sample region: 1 = South Africa, -1 = Botswana;
		RC = replicative capacity; VL = viral load; CD4 = CD4 cell count.}
	\label{fig:3}
\end{figure}
We identify one escape mutation I147X being significantly more prevalent in Botswana as indicated by its negative correlation with South Africa. Located at the amino-terminal position of an HLA-B57-restricted epitope (`ISW9'), variation at \textit{gag} residue 147  is known to be associated with expression of B57 \citep{Draenert:2004sw}. 
It is worth noting that three of the four escape mutations that correlate negatively with RC (I61X, Q182X and T242X) have a higher frequency in Botswana and may therefore have contributed to the lower RC found in Botswana by \citet{Payne2014}.
Interestingly, the negative effect on RC we estimate for two mutations finds clear confirmation in experimental testing: in vitro experiments provide evidence for a reduction in RC by T242X \citep{martinez2006fitness,song2012impact} and T186X is also found to greatly impair RC \citep{huang2011progression}.

Our analysis recovers the expected inverse correlation between CD4 count and RC or VL, as well as the positive correlation between RC and VL \citep{prince2012role}, confirming that more virulent viruses result in faster disease progression.
Also, South Africa is associated with higher VL and lower CD4, suggesting that the South African cohort may comprise individuals with more advanced disease, even though the two cohorts are closely matched in age \citep{Payne2014}. This is somewhat at odds with the original study that also finds a higher VL for South Africa, but at the same time a higher CD4 count for patients from this country. Such differences are likely to arise from controlling or not for the phylogeny.

The remaining significant correlation between escape mutations (row 1 to 19 in Figure \ref{fig:3}) can be considered as epistatic interactions, some of which are strongly positive.
For example, we find a strong positive correlation between T186X and T190X. The former represents an escape mutation for HLA-B*81-mediated immune responses and has been reported to be strongly correlated with reduced virus replication \citep{huang2011progression,Wright:2010en},
as also reflected in the negative correlation between this mutation and RC.
In fact, \cite{wright2012impact} show T186X requires T190I (or Q182X, also positively correlated with T186X, Figure \ref{fig:3}) to partly compensate for this impaired RC.
The other strong positive correlation between A163X and S165X has also been found to be a case of a compensatory mutation, with S165N partially compensating for the reduced viral RC of A163G \citep{Crawford:2007wn}.
The same holds true for the positive correlation between A146X and I147X, with I147L partially compensating the fitness cost associated with the escape mutation A146P \citep{Troyer:2009ol}.
\subsection{Tree inference}\label{sec:application_tree}
Figure \ref{fig:tree} reports the maximum clade credibility tree from the posterior sample. The tree maximizes the sum of posterior clade probabilities. The posterior mean tree height is roughly 30 years; so with the most recent samples from 2010, we date the common ancestor of all viruses back to around 1980, consistent with the beginning of this epidemic.
\begin{figure}[ht!]
	\centering
	\includegraphics[scale=0.4]{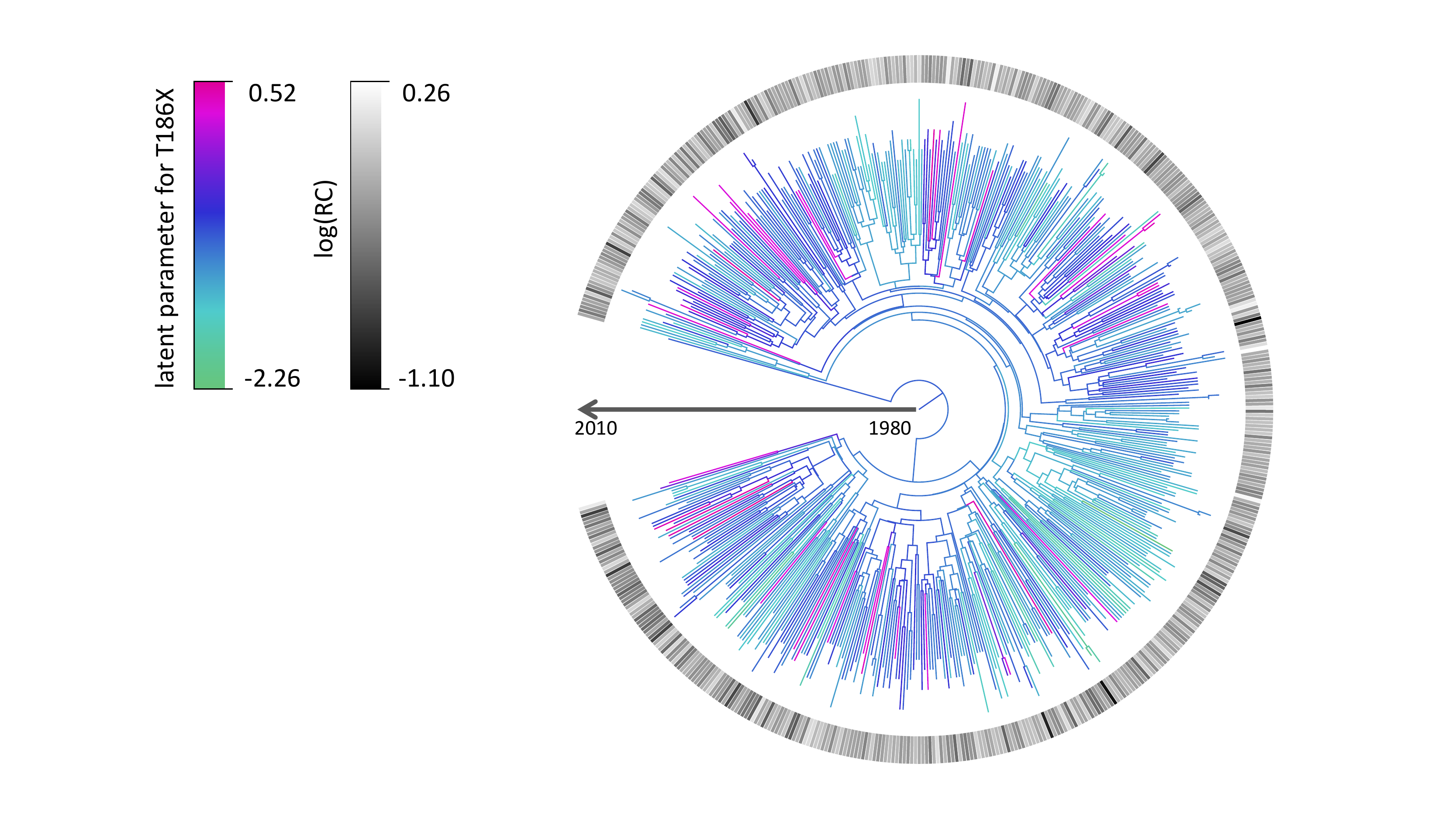}
	\caption{The maximum clade credibility tree with branches colored by the posterior mean of the latent parameter corresponding to mutation T186X. Outer circle shows log(RC) in gray scale.}
	\label{fig:tree}
\end{figure}

\section{Efficiency comparison and goodness-of-fit test}
\subsection{Efficiency comparison}\label{sec:compareEfficiency}
To compare efficiency of BPS with the multiple-try rejection sampling in \citet{Cybis2015}, we run both samplers on the whole data set ($ \nTaxa = 535, \nTraits = 24 $) and a subset with $ \nTraits = 8$ including the three continuous traits, and fix the tree and across-trait covariance at the same values from preliminary runs.
The efficiency criterion is per unit-time ESS across all $ \nTaxa\nTraits $ latent parameters.
BPS outperforms rejection sampling to a greater extent as $ \nTraits $ increases. For $\nTraits = 24$, BPS yields a $74\times$ increase in terms of the minimum ESS and an $11\times$ increase for the median ESS (Table \ref{tb:ess}).
This order-of-magnitude improvement is more clear in Figure \ref{fig:histogramSubsetP}.
Because rejection sampling only updates one taxon per iteration, some latent parameters rarely change their values (Figure \ref{fig:trace}).
As a result, the minimum ESS of multiple-try rejection sampling is much lower than BPS which simultaneously updates all latent dimensions.
\mytableEffiency
\begin{figure}[h!]
	\centering
	\includegraphics[scale=0.5]{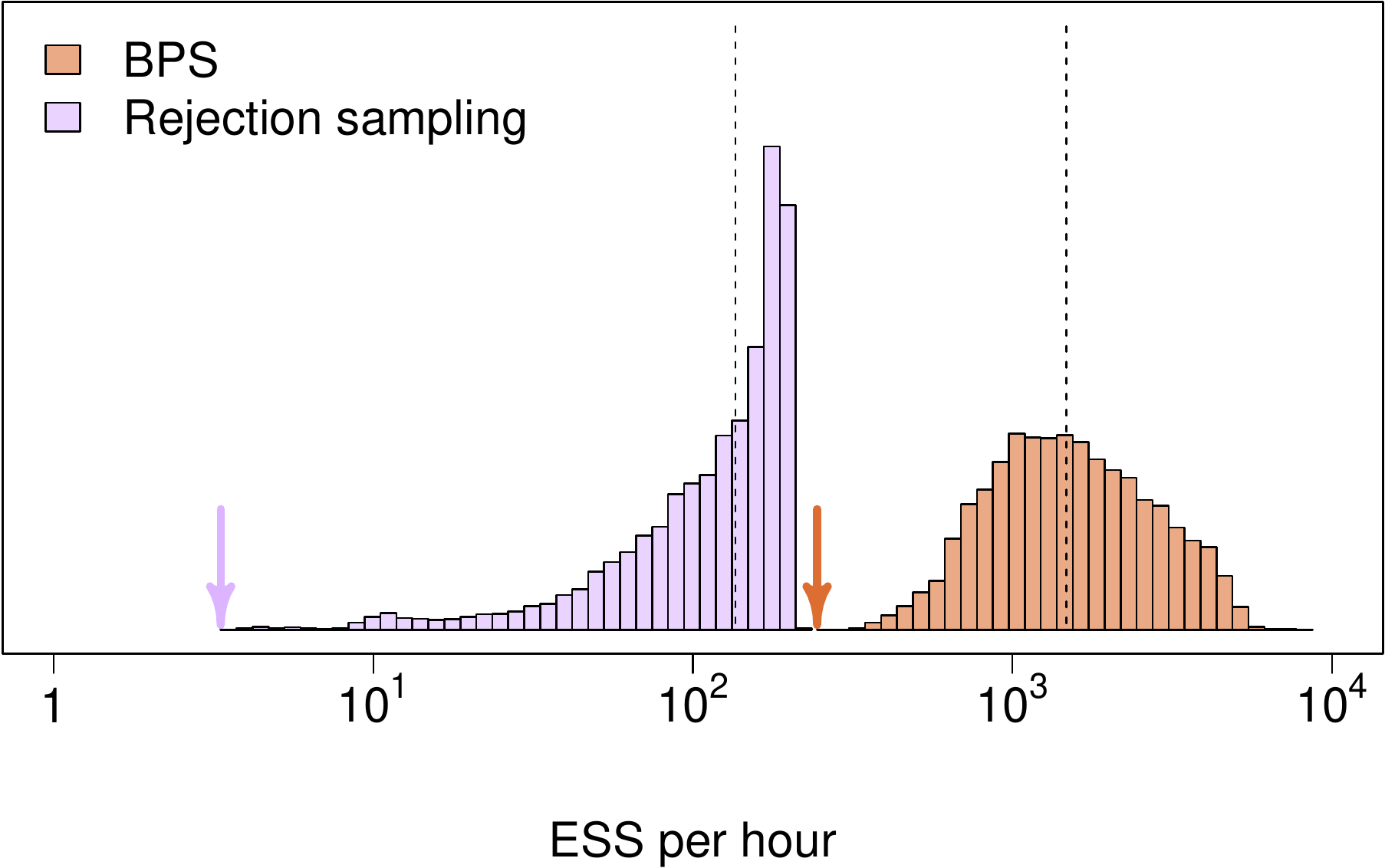}
	\caption{A representative histogram of ESS across latent parameters, sampled by BPS or rejection sampling in one hour run-time. Arrows and dashed lines denote the minimum and median ESS ($\nTaxa = 535, \nTraits = 24 $).}
	\label{fig:histogramSubsetP}
\end{figure}
\begin{figure}[h!]
	\centering
	\includegraphics[scale=0.45]{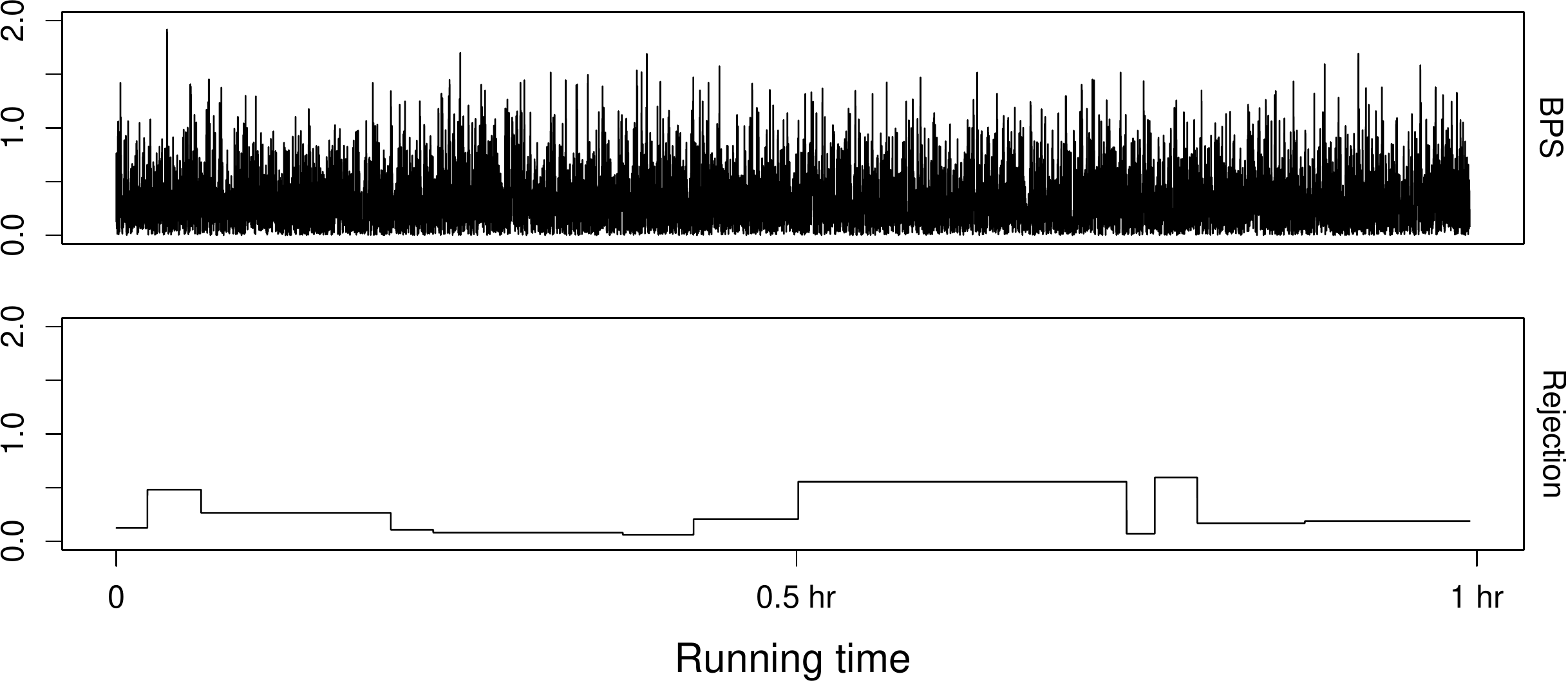}
	\caption{Trace plot of the latent parameter with the least ESS by rejection sampling (bottom) and trace plot of the same latent parameter sampled by BPS (top) for an one hour run-time. BPS and rejection sampling run $ 1.1 \times 10^4 $ and $ 2.6 \times 10^5 $ iterations, respectively ($ \nTaxa  = 535, \nTraits = 24$).}
	\label{fig:trace}
\end{figure}

\subsection{Model goodness-of-fit}\label{sec:goodness}
We compare the phylogenetic probit model fit to reduced models that do not include phylogenetic correction.
This comparison not only allows us to assess goodness-of-fit of the phylogenetic probit model, but also tests whether explicit tree modeling is necessary in practice.
The two reduced models  both assume independence among virus samples such that the across-taxa tree covariance $ \phylogenyVariance$ is diagonal.
The first ``dated star'' model incorporates varying viral sampling time information
such that $ \phylogenyVariance$ has diagonal elements equal to the time distance from virus sample date to the root date fixed, without loss of generality, to 1980.
To understand the star-moniker, phylogeneticists often use a ``star-tree'' in which all branch lengths between internal nodes equal $0$ to represent independent samples.
The second ``ultrametric star'' model, assumes that all taxa have traits that are identically distributed so $\phylogenyVariance$ is an identity matrix. 

For each of the three models, we assess out-of-sample prediction by repeatedly splitting up the HIV data into a training set used to build each model, and a test set to evaluate the prediction.
Across the 21 binary traits for all taxa, we hold out $\numTestTraits = 21 \times 535 \times 20\%$ of the observations, build the model and then estimate the posterior probability $\probEstimate{h}$ for $h = 1,\ldots \numTestTraits$ that held-out trait $h$ equals its observed value.

We summarize performance through quantiles of the score $\log \probEstimate{h}$ to measure accuracy, and a higher score represents better prediction  (Table \ref{tb:prediction}). 
The phylogenetic probit model commands higher scores compared to the two reduced models and we conclude that joint tree modeling through the phylogenetic probit model leads to better data fit. 

\mytableGoodness

\section{Discussion}
We present an efficient Bayesian inference framework to learn the correlation among mixed-type traits across a large number of taxa, while jointly inferring the phylogenetic tree through sequence data.
Our approach significantly improves upon \cite{Cybis2015} in both modeling and inference. Better modeling comes from the decomposition of across-trait covariance matrix $ \traitCovariance= \traitDiag \traitCorr  \traitDiag$ that keeps the generalized probit model identifiable and allows a jointly uniform LKJ prior on $\traitCorr$.
Compared to the convenient but restrictive Wishart prior that causes mixing problems for sampling  $ \traitPrecision $ and $\latentData$, this decomposition facilitates correlation inference among continuous traits and latent parameters (Appendix Figure \ref{sup:tracePlot}). Our main contribution lies in an efficient inference framework, specifically, an optimized BPS to sample latent parameters from a high-dimensional truncated normal distribution.
In contrast to the ``one-taxon-at-a-time" design in \citet{Cybis2015}, BPS jointly updates all dimensions therefore reducing auto-correlation among MCMC samples. The most expensive steps involved are matrix-vector multiplications by the precision matrix $ \bpsPrecision =\grandPrecision$.
In our case, the tree precision matrix is unknown and getting it by matrix inversion is notoriously $ \order{\nTaxa^3} $. Thanks to the insight in Proposition \ref{proposition1}, we circumvent this obstacle by utilizing a dynamic programming strategy and obtain the desired matrix-vector products in $ \order{\nTaxa\nTraits^2}$. BPS also enjoys an advantage especially important for mixed-type traits. That is, we can simply ``mask out" the fixed continuous traits when sampling latent parameters for binary traits. Whereas the rejection sampling in \citet{Cybis2015} has to calculate the conditional distribution of latent dimensions given continuous traits at each tip. This cost-free ``masking"  technique to condition on a subset of dimensions exploits properties of normal distributions and can be shared with other dynamics-based sampler, like HMC.
Taking all of these points together, the optimized BPS provides a huge gain in efficiency.

Naturally, BPS may also be an efficient choice in situations where $ \bpsPrecision $ itself has special structures that facilitate quick matrix-vector multiplication. For example, inducing precision matrices that are sparse or composed of sparse components is a common strategy for analyzing large spatial data \citep{heaton2019case}. 
Methods like the nearest neighbor Gaussian process \citep{datta2016hierarchical}, integrated nested Laplace approximations \citep{rue2009approximate}, and multi-resolution approximation of Gaussian processes \citep{katzfuss2017multi} all achieve computational efficiency from sparsity in $ \bpsPrecision $.
Whether BPS would be useful in these scenarios, especially with mixed-type data, is an interesting topic for future research.

Our application provides important information on the complex association between HLA-driven HIV \textit{gag} mutations and virulence that was previously assessed by experimental and epidemiological studies. To our best knowledge, this is the first study to examine essential HIV virus-host interactions while explicitly modeling the phylogenetic tree.
Our setup is also different from the original study \citep{Payne2014} in that we attempt to identify correlation between individual epitope escape mutations, virulence, and country of sampling, instead of considering all mutations together or grouping them with particular HLA types (e.g.~HLA-B*57/58:01). While the latter may increase power to detect population-level differences in escape mutation frequencies, our approach allows us to pinpoint particular mutations contributing to virulence.
Good consistency between the mutations that we associate with reduced RC and literature reports on virological assays suggests that our approach may complement or help in prioritizing experimental testing, and therefore further assist in the battle against HIV-1. Our method contributes to a general framework to assess correlation among mixed-type traits in virology, but also more broadly in evolutionary biology.

One future improvement lies in the prior choice on across-trait correlation.
The LKJ prior works well for our $ \nTaxa = 535, \nTraits = 24 $ data set, as it is noninformative as desired, and correlation elements are well-mixed through No-U-Turn HMC.
Under this choice, we view correlations with 90\% HPD intervals not covering zero as significant.
We can adjust this decision threshold based on resource availability for follow-up experimental studies.
However, with much larger $ \nTraits $ and when only a small portion of the observed traits are truly involved in the underlying biology, it becomes vital to control for false positive signals, and one may favor a systematic solution.
For example, it may be preferable to put a shrinkage-based prior on $ \traitCorr $ that shrinks individual elements towards zero. Ideas like the graphical lasso prior \citep{wang2012bayesian} and  factor models with shrinkage prior on the loading elements \citep{bhattacharya2011sparse} are potential directions to explore.

Lastly, as understanding the relationship among mixed-type variables is a common question in different fields, our method suits a large class of problems beyond evolutionary biology. The optimized BPS sampler through dynamic programming serves as an efficient inference tool for any multilevel (hierarchical) model \citep{gelman2006multilevel} with an additive covariance structure on a directed acyclic graph (Figure \ref{fig:vmatrix}). The tree variance matrix $\phylogenyVariance$ that we use to describe the covariation of shared evolutionary history also arises from other kinds of relationships.
For example, additive covariance includes pedigree-based or genomic relationship matrices in animal breeding \citep{vitezica2013additive, mrode2014linear} and distance matrices decided by geographical locations in infectious disease research \citep{barbu2013effects}. Intriguingly, our dynamic programming strategy also provides a way to invert the $ \nTaxa \times \nTaxa$ tree variance matrix $\phylogenyVariance $ in $ \order{\nTaxa^2} $ by piecing together the products $\phylogenyPrecision \basisVector{i}$ for $i = 1, \ldots, \nTaxa$.
While this seems likely a well-known result, we have failed to find precedence in the literature.
Finally, the phylogenetic probit model can be generalized to categorical and ordinal data, which will only add to its broad applicability.

\section{Acknowledgments}
We thank Oliver Pybus for useful discussions on an earlier version of the data set analyzed here. The research leading to these results has received funding from the European Research Council under the European Union's Horizon 2020 research and innovation programme (grant agreement no.~725422 - ReservoirDOCS).
The Artic Network receives funding from the Wellcome Trust through project 206298/Z/17/Z.
PB acknowledges support by the Research Foundation -- Flanders (`Fonds voor Wetenschappelijk Onderzoek – Vlaanderen', 12Q5619N and V434319N).
MAS acknowledges support through NSF grant DMS 1264153 and NIH grants R01 AI107034 and U19 AI135995.
PL acknowledges support by the Research Foundation -- Flanders (`Fonds voor Wetenschappelijk Onderzoek -- Vlaanderen', G066215N, G0D5117N and G0B9317N).

\renewcommand{\thefigure}{A.\arabic{figure}}
\renewcommand{\thetable}{A.\arabic{table}}
\setcounter{table}{0}
\setcounter{figure}{1}
\appendix
\section{BPS details}
\subsection{BPS modification for conditional truncated MVNs}
\label{sec:conditional tMVN}
Here we consider modifying the BPS to incorporate fixed dimensions that are the observed, continuous traits in our mixed-type model. We partition $ \bpsPosition = \left( \bpsPosition_b , \bpsPosition_c \right)$
by latent and observed dimensions and then generate samples from the conditional distribution $ \cDensity{\bpsPosition_b}{\bpsPosition_c}$. To make progress, we parameterize $ \cDensity{\bpsPosition_b}{\bpsPosition_c}$ in terms of $ \density{\bpsPosition} $ with partitioned mean $  \tipMeanMatrixVec = (\tipMeanMatrixVec_b, \tipMeanMatrixVec_c) $ and precision matrix
\begin{equation}
\grandPrecision =
\begin{bmatrix}
\bPhi_{bb} & \bPhi_{bc} \\
\bPhi_{cb} & \bPhi_{cc}
\end{bmatrix}.
\end{equation}
With a similarly partitioned velocity $ \bpsVelocity = (\bpsVelocity_b , \bpsVelocity_c) $, the distribution $ \cDensity{\bpsPosition_b}{\bpsPosition_c}$ carries potential energy
\begin{equation}
U_{b \given c}(\bpsPosition_b + t \bpsVelocity_b)
= \frac{t^2}{2} \bpsVelocity_b^\intercal \bPhi_{bb} \bpsVelocity_b + t \bpsVelocity_b^\intercal \bPhi_{bb} (\bpsPosition_b - \tipMeanMatrixVec_{b \given c}) + \constant,
\end{equation}
where constant $ \constant $ does not depend on $ t $. The conditional mean  $\tipMeanMatrixVec_{b \given c} = \tipMeanMatrixVec_b - \bPhi_{bb}^{-1} \bPhi_{bc} (\bpsPosition_c - \tipMeanMatrixVec_c)$, so
\begin{multline}\label{eq:energy_masking}
U_{b \given c}(\bpsPosition_b + t \bpsVelocity_b)
\\= \frac{t^2}{2} \bpsVelocity_b^\intercal \bPhi_{bb} \bpsVelocity_b + t \bpsVelocity_b^\intercal \left[ \bPhi_{bb} (\bpsPosition_b - \tipMeanMatrixVec_b) + \bPhi_{bc} (\bpsPosition_c - \tipMeanMatrixVec_c) \right]
+ \constant.
\end{multline}
This expression is equivalent to masking out the dimensions of $ \bpsVelocity $ in \eqref{eq:truncated_normal_energy} that corresponds to $ \bpsPosition_c $ via the vector $\tilde{\bpsVelocity} = (\bpsVelocity_b, \bm{0})$.  To be explicit, we rewrite \eqref{eq:energy_masking} as
\begin{equation}
U_{b \given c}(\bpsPosition_b + t \bpsVelocity_b)
= \frac{t^2}{2} \tilde{\bpsVelocity}^\intercal \bPhi \tilde{\bpsVelocity} + t \tilde{\bpsVelocity}^\intercal \bPhi (\bpsPosition - \tipMeanMatrixVec)
+ \constant.
\end{equation}
Therefore, adding this masking operation for $ \bpsVelocity,\bpsPhix,\bpsPhiv $ in Lines 1, 2, 5, 7 in Algorithm~\ref{alg:bps_conditional} allows sampling
from the conditional truncated MVN $\cDensity{\bpsPosition_b}{\bpsPosition_c}$ without any additional cost.
\subsection{Tuning $\totalTime$ for BPS}\label{sec:tuningBPS}
The total simulation time $ \totalTime $ for the Markov process is a tuning parameter in Algorithm~\ref{alg:bps_conditional}.
	If $\totalTime$ is too small, the particle does not travel far enough from the initial position, leading to high auto-correlation among MCMC samples.
	On the other hand, an unnecessarily large $\totalTime$ would waste computational efforts without any substantial gain in mixing rate. 
	To achieve best computational efficiency, therefore, one would like to choose a $\totalTime$ just large enough that $\bpsPosition(\totalTime)$ is effectively independent of $\bpsPosition(0)$.
	To help find such $\totalTime$ for BPS applied to truncated MVNs, we develop a heuristic based on the following observations.
	
	At stationarity, the BPS has a velocity distributed as $\normalDistribution{\bm{0}}{\I}$.
	In other words, we have $\bpsVelocity(t) \sim \normalDistribution{\bm{0}}{\I}$ for all $t \geq 0$ if starting from stationarity.
	In particular, the velocity along any unit vector $\bu$ would be distributed as $\langle \bpsVelocity(t), \bu \rangle \sim \normalDistribution{0}{1}$, so that $\mathbb{E}|\langle \bpsVelocity(t), \bu \rangle| = \sqrt{2 / \pi}$. 
	Now, the motion of the particle along $\bu$ is given by $\langle \bpsPosition(t), \bu \rangle = \langle \bpsPosition(0), \bu \rangle + \int_0^t \langle \bpsVelocity(s), \bu \rangle \mathrm{d} s$.
	At the same time, for a MVN with covariance $\grandVariance$, its high density region has a diameter proportional to $\sqrt{\lambda_{\text{max}}}$, where $\lambda_{\text{max}}$ denotes the largest eigenvalue of $ \grandVariance $.
	Therefore, in order to allow the particle to travel across the high density region, we would like it to move a distance proportional to $\sqrt{\lambda_{\text{max}}}$, that is, $ | \int_0^{\totalTime} \langle \bpsVelocity(s), \bu \rangle \diff s | \propto \sqrt{\lambda_{\text{max}}} $.
	
	Since BPS is designed to suppress the random-walk behavior of more traditional MCMC algorithms \citep{peters2012}, we expect the motion of the particle along $\bu$ not to change its direction frequently.
	Or equivalently, we expect the velocity along $\bu$, given by $ \langle \bpsVelocity(t), \bu \rangle $, not to change its sign frequently.
	When there is no change in $\langle \bpsVelocity(t), \bu \rangle$ during $[0, \totalTime]$, we would have
	$  | \int_0^{\totalTime} \langle \bpsVelocity(s), \bu \rangle \diff s |  = \int_0^{\totalTime} | \langle \bpsVelocity(s), \bu \rangle | \diff s $.
	This, combined with the observation that $\mathbb{E}|\langle \bpsVelocity(t), \bu \rangle| = \sqrt{2 / \pi}$ at stationarity, suggest that roughly, the particle moves an average distance of $ \sqrt{2 / \pi} $ during one unit of time. 
	We so conjecture that there is a choice of travel time $\totalTime \propto \sqrt{\lambda_{\text{max}}}$ that achieves $| \int_0^{\totalTime} \langle \bpsVelocity(s), \bu \rangle \diff s | \propto \sqrt{\lambda_{\text{max}}} $ and good mixing.
	This heuristic applies to a truncated MVN when assuming its high density region diameter is comparable to that of the untruncated MVN.
	We find that BPS performance is not overly sensitive to a specific choice of $\totalTime$.
	After preliminary runs (Table \ref{sup:tableForTravelTime}), we choose $ \totalTime = 0.01\sqrt{\grandVarianceEigenvalueLargest} $ for our $ \nTaxa = 535, \nTraits = 24 $ application, as it yields the maximum median effective sample size (ESS) per hour run-time.

\begin{table}[h]
	\caption{Effective sample size per hour run-time (ESS/hr) of latent parameters sampled by BPS with different $ \totalTime $. We fix the tree and use the No-U-Turn sampler to sample the across-trait covariance matrix. With $ \totalTime = 0.01 \sqrt{\grandVarianceEigenvalueLargest} $, the minimum, 5\%, and 50\% percentile of ESS/hr are either larger or close to those with other $ \totalTime $ values compared. } \label{sup:tableForTravelTime}
	\centering
	\begin{tabular}{lrrrr}
		\toprule
		\multicolumn{1}{c}{ }& \multicolumn{4}{c}{$ \totalTime $} \\
		\cmidrule(l{2pt}r{2pt}){2-4}
		ESS/hr percentile& $ 5 \times 10^{-3} \sqrt{\grandVarianceEigenvalueLargest}  $ & $  10^{-2} \sqrt{\grandVarianceEigenvalueLargest}  $ & $  10^{-1} \sqrt{\grandVarianceEigenvalueLargest}  $& \\
		\midrule
min &  72 & 68 & 27\\
5\% &  227 & 428 & 357\\
50\% &  515 & 1050 & 885\\
		\bottomrule
	\end{tabular}
\end{table}\label{sec:tunint_bps_travel_time}
\section{Identifiability issue with a Wishart prior}
We examine differences between assuming an LKJ + log normal priors on $  \traitDiag \traitCorr  \traitDiag $ and a Wishart prior on $ \traitPrecision $.
For the Wishart case, we set the degree of freedom equal to $ \nTraits + 1 $, so each correlation marginally follows a uniform distribution on $ \left[-1, 1\right] $ \citep{bda13}, and the Normal-Wishart conjugacy yields easy Gibbs sampling for $ \traitPrecision $. Without constraining the marginal variance of any latent dimension, the Wishart prior leaves the model not parameter-identifiable and causes mixing problems, even with a small $ \nTraits = 8$ (Figure \ref{sup:tracePlot}).

\begin{figure}[h!]
	\begin{subfigure}{\textwidth}
		\centering
		\includegraphics[scale=0.5]{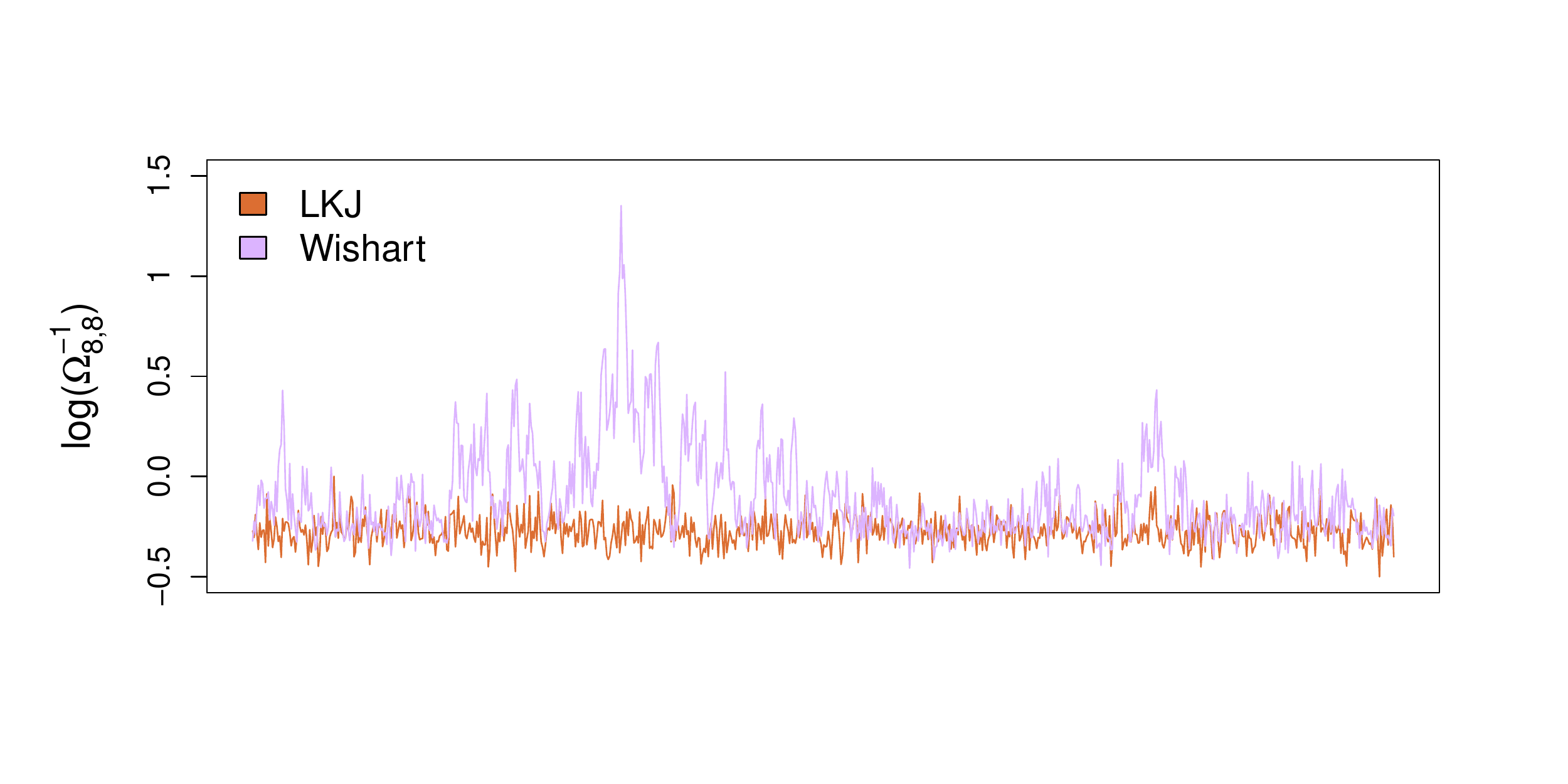}\vspace{-5em}
	\end{subfigure}
	\begin{subfigure}{\textwidth}
		\centering
		\includegraphics[scale=0.5]{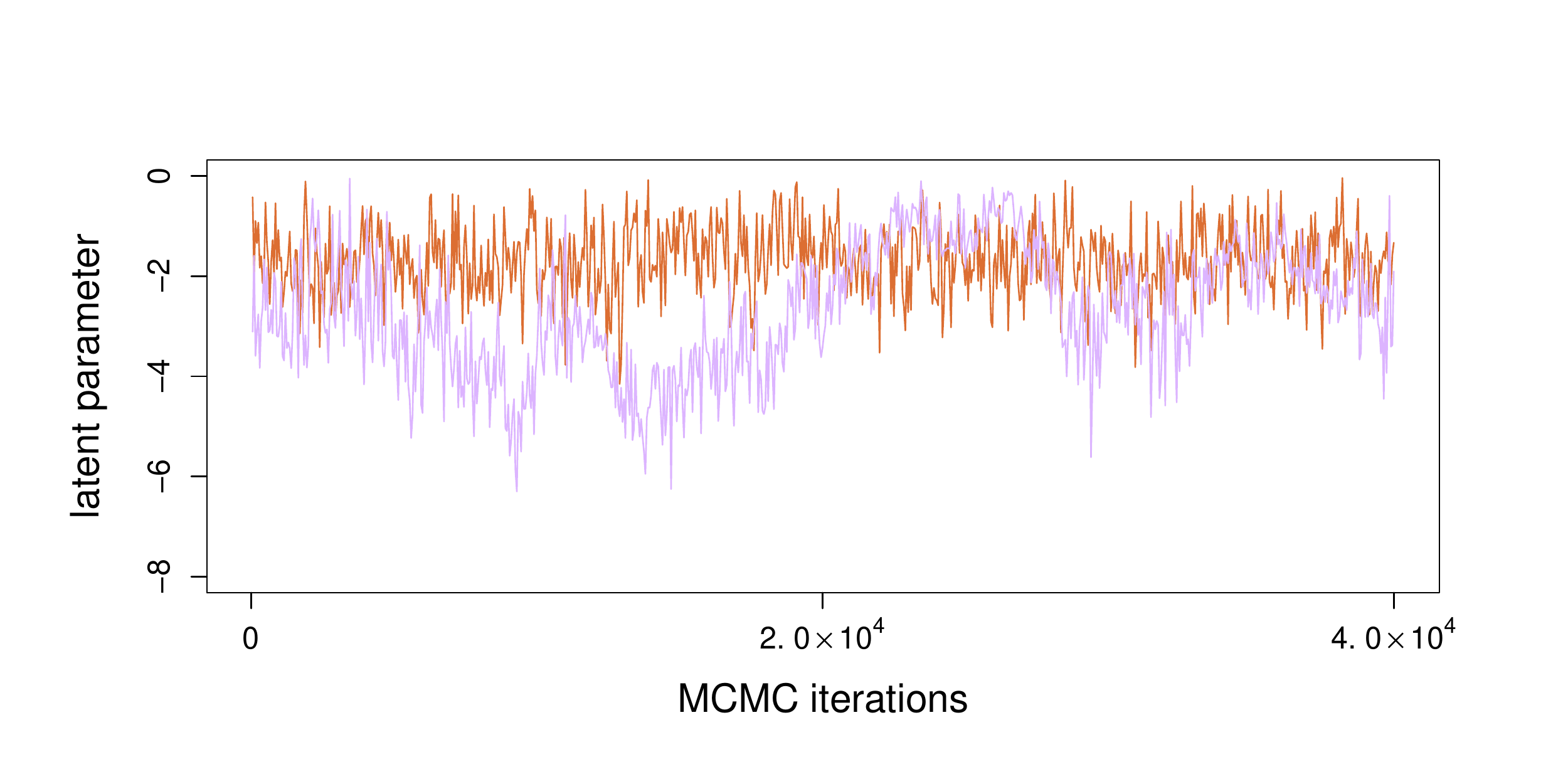}
	\end{subfigure}
	\caption{Trace plot of a representative $ \traitPrecision $ element (top) in log scale and the latent parameter with the least ESS when assuming a Wishart prior on $ \traitPrecision $ (bottom).}
	\label{sup:tracePlot}
\end{figure}

\begin{supplement}
 \stitle{Data set and source code}
 \slink[doi]{COMPLETED BY THE TYPESETTER}
 \sdatatype{.zip}
	\sdescription{We provide the HIV data set and source code to reproduce results in the article.}
\end{supplement}
\clearpage
\bibliographystyle{imsart-nameyear}
\bibliography{bps_paper}
\end{document}